\documentclass[11pt]{article}
\usepackage[T1]{fontenc}
\usepackage{lmodern}
\usepackage[margin=1in]{geometry}
\usepackage{amsmath,amssymb,amsthm,booktabs,longtable,array}
\usepackage{microtype}
\usepackage{xurl}
\usepackage[hidelinks]{hyperref}
\newtheorem{theorem}{Theorem}[section]

\newtheorem{lemma}[theorem]{Lemma}
\newtheorem{proposition}[theorem]{Proposition}
\newtheorem{corollary}[theorem]{Corollary}
\theoremstyle{definition}
\newtheorem{definition}{Definition}[section]
\theoremstyle{remark}
\newcommand{\cc}{\mathbin{\boxtimes}}
\newcommand{\gd}{\mathsf e}
\newcommand{\gap}{\Delta_{\mathrm{fp}}}
\newcommand{\DE}{\mathsf T_c}
\newcommand{\Fneq}{\mathcal F_{\ne\gd}}
\newcommand{\Fnt}{\mathcal F_{\mathrm{nt}}}
\newcommand{\E}{\mathbb E}
\newcommand{\Rk}{\mathcal R}
\DeclareMathOperator{\rank}{rank}
\title{Spatially Coupled MacKay--Neal Codes Achieve Capacity on BMS Channels at Fixed Degrees}
\author{Kenta Kasai\\Institute of Science Tokyo}
\date{}

\hypersetup{pdftitle={Spatially Coupled MacKay--Neal Codes Achieve Capacity on BMS Channels at Fixed Degrees},pdfauthor={Kenta Kasai},pdfsubject={MN codes, fixed degrees, sum-product decoding},pdfkeywords={MacKay-Neal codes, spatial coupling, threshold saturation, symmetric channels}}
\begin{document}
\maketitle
\begin{abstract}
We establish two degree-dependent results for spatially coupled MacKay--Neal codes on binary-input memoryless symmetric channels. The ensembles use uniform random smoothing and shorten both variable types outside the active chain. For $r=g=3$, every integer $\ell\geq4$, and each channel of capacity greater than $R=r/\ell$, there is a sequence of code realizations whose actual transmitted rate tends to $R$ and whose average sum-product bit error tends to zero. For $r=g=2$ and each $\ell\in\{3,4,5\}$, an interval of binary symmetric channels has capacity greater than $R$ but retains positive transmitted-bit error under terminated density evolution as chain length grows relative to coupling width. Both results follow from the signs of the same density-valued potential at uncoupled fixed points. The degree-three proof combines analytic bounds for $\ell\geq33$ with exact interval certificates for $4\leq\ell\leq32$, followed by threshold saturation and a projection argument for the actual rate. The degree-two proof analytically constructs a fixed point with negative potential and controls both boundary contributions. We relate the potential and the degree-two bifurcation condition to earlier statistical-mechanical predictions. The finite certificates and verification software are available in a versioned supplement.
\end{abstract}

\section{Introduction}
Spatially coupled low-density parity-check (LDPC) codes developed from constructions originally studied as LDPC convolutional codes. Time-varying periodic convolutional codes defined by sparse parity-check matrices were introduced in 1999 together with iterative decoding~\cite{FelstromZigangirov1999}. Subsequent constructions related quasi-cyclic LDPC block codes to time-invariant convolutional codes through circulant matrices and demonstrated improved iterative-decoding performance for the convolutional codes~\cite{TannerEtAl2004}. Density-evolution studies of terminated regular LDPC convolutional ensembles then showed substantial threshold improvements over the corresponding block ensembles on the binary erasure channel (BEC) and the binary-input additive white Gaussian noise (AWGN) channel~\cite{LentmaierEtAl2010}. Thus, the favorable decoding behavior was observed and analyzed under the convolutional-code terminology before its explanation through spatial coupling.

The spatial-coupling viewpoint describes this structure as a chain of sparse graphs connected across neighboring positions. The BEC threshold-saturation theorem, published in 2011, explained how coupling and termination can raise the sum-product threshold to the maximum-a-posteriori (MAP) threshold of the underlying uncoupled ensemble in the coupling limit~\cite{KudekarRichardsonUrbanke2011}. This mechanism was subsequently developed into threshold saturation and universal capacity achievement over binary-input memoryless symmetric (BMS) channels for spatially coupled regular LDPC ensembles~\cite{KudekarRichardsonUrbanke2013}. Approaching capacity in that construction allows the node degrees to increase.

A different question is whether one can approach capacity while keeping the degrees fixed. MacKay--Neal (MN) codes arose from the sparse-matrix constructions introduced in 1995 and developed further in 1999~\cite{MacKayNeal1995,MacKay1999}. They combine punctured variables, whose bits are not transmitted, with transmitted variables in a sparse two-type graph. Statistical-mechanical analyses subsequently predicted capacity-achieving performance under optimal decoding for suitable regular MN degrees on symmetric channels, using the replica-symmetric free energy~\cite{TanakaSaad2003}. The relation between that free energy and the potential used here is given by~\eqref{eq:rsU} in Section~\ref{sec:potential}.

A related development was the construction of Hsu--Anastasopoulos (HA) codes by concatenating an outer LDPC code with an inner low-density generator-matrix (LDGM) code. The capacity theorem published in 2010 established that such constructions can achieve capacity on BMS channels under maximum-likelihood (ML) decoding with bounded graphical complexity, measured by the number of graph edges per information bit~\cite{HsuAnastasopoulos2010}. MN and HA constructions are related by code duality; their sparse representations retain punctured variables even when eliminating those variables yields dense matrices on the transmitted coordinates~\cite{KasaiSakaniwa2011}. These results motivated the search for bounded-degree constructions with capacity-achieving iterative decoding. Throughout this paper, belief propagation (BP) denotes the sum-product algorithm (SPA).

Spatially coupled MN and HA codes were introduced in 2011, and density evolution gave numerical evidence of thresholds close to capacity on the BEC~\cite{KasaiSakaniwa2011}. A bounded-degree BEC capacity result was established for the $(k,2,2)$ MN family~\cite{ObataJianKasaiPfister2013}. For the $(\ell,3,3)$ family, capacity achievement was proved at fixed degrees by establishing positivity of the uncoupled potential at the relevant fixed points and applying threshold saturation~\cite{OkazakiKasai2014}. The analysis was extended to generalized erasure channels with memory~\cite{FukushimaOkazakiKasai2015}. 

For the binary-input AWGN channel, thresholds close to capacity were estimated for spatially coupled MN protographs using the reciprocal-channel approximation~\cite{MitchellKasaiLentmaierCostello2012}. More recently, near-capacity thresholds have been studied for rate-adaptive constructions that concatenate an outer distribution matcher with a spatially coupled inner code~\cite{ZahrLiva2025}. 

The MN/HA structure has also been applied to quantum coding. Nested Calderbank--Shor--Steane (CSS) code pairs constructed from MN and HA codes have been shown to have positive asymptotic rate and relative distances bounded away from zero with fixed degrees in their punctured sparse representations; specified choices of these degrees attain the CSS Gilbert--Varshamov (GV) distance bound~\cite{KasaiQuantumGV2026}. A second study derived a five-message density-evolution recursion for spatially coupled MN/HA CSS codes on the quantum erasure channel. With an ideal decoding seed and equal constituent design rates on the $X$ and $Z$ sides, its threshold-saturation analysis establishes the hashing-bound threshold at the density-evolution level~\cite{KasaiQuantumErasure2026}. These works connect the MN/HA structure to quantum distance bounds and iterative-decoding thresholds, respectively.

The present paper focuses on classical communication over BMS channels with the uniformly smoothed regular MN ensemble. Its proof follows the BEC strategy of evaluating trivial fixed points explicitly and controlling the potential at all remaining fixed points. This distinction leads to capacity achievement at degree three and to the degree-two obstruction stated below.

\subsection{Two degree regimes}
In the $(\ell,r,g)$ notation, punctured variables have degree $\ell$, transmitted variables have degree $g$, and each check has $r$ sockets for punctured variables and $g$ for transmitted variables. The design rate $R=r/\ell$ is obtained by counting variables and check equations before termination, treating the equations as independent and assuming no dimension loss under puncturing. The denominator counts only transmitted bits. Section~\ref{sec:rate} proves that the actual transmitted rate of the selected code sequence tends to $R$. Section~\ref{sec:model} defines the uniformly smoothed ensemble, with both variable types shortened outside the active chain.

\begin{theorem}[Capacity achievement at degree three]\label{thm:main}
Set $r=g=3$, fix an integer $\ell\geq4$, and let $R=r/\ell$. On every BMS channel $c$ whose capacity satisfies $C(c)>R$, the terminated, uniformly smoothed spatially coupled $(\ell,r,g)$-MN ensembles admit a sequence of code realizations whose actual transmitted rate tends to $R$ and whose average BP bit error probabilities for both punctured and transmitted bits tend to zero. The degree parameters $\ell,r,g$ remain fixed throughout the sequence.

For every $\delta>0$ with $R+\delta\leq1$, a common finite coupling width suffices for density evolution to converge to perfect information on every finite active chain and every BMS channel with $C(c)\geq R+\delta$.
\end{theorem}
The theorem achieves the Shannon boundary at the discrete rates $r/\ell$. Its code realizations may depend on the target channel. Corollary~\ref{prop:widthmain} gives a sufficient width in terms of a positive potential gap; Section~\ref{sec:rate} proves the actual-rate and finite-code claims. The common-width statement concerns the density-evolution limit, in which the section size tends to infinity before the iteration count.

\begin{theorem}[A degree-two obstruction above the design rate]\label{thm:degreetwo}
Set $r=g=2$ and $R=r/\ell$. For each $\ell\in\{3,4,5\}$, there is a nonempty open interval of BSC crossover probabilities whose channels $c$ satisfy $C(c)>R$ and the following property. Writing $\overline P_{\mathrm{tx}}(L,w;c)$ for the limiting average transmitted-bit error of terminated density evolution initialized with channel observations, every sequence $L/w\to\infty$ satisfies
\[
\liminf\overline P_{\mathrm{tx}}(L,w;c)>0.
\]
The coupling width may grow along the sequence.
\end{theorem}
Section~\ref{sec:degreetwo} proves this result by constructing an uncoupled fixed point of negative potential and bounding the effect of both shortened boundaries. The density-evolution limit takes the section size to infinity at each fixed iteration count, then takes the iteration limit for fixed $L,w$. Theorem~\ref{thm:degreetwo} gives an obstruction in this order of limits. Its conclusion concerns the specified ensemble and decoding recursion; the degree-two MAP threshold remains a separate question.

These two theorems distinguish degrees three and two on general symmetric channels, although both families achieve the BEC capacity. The statistical-mechanical analysis of uncoupled MN codes already predicts the three BSC exceptions in Theorem~\ref{thm:degreetwo}, together with capacity-boundary transitions for the degree-three regime~\cite{TanakaSaad2003}. The parameters in that analysis satisfy $(K,C,L)=(r,\ell,g)$. The potential in Section~\ref{sec:potential} is an explicit normalization of the replica-symmetric free-energy expression in that analysis, and Section~\ref{sec:degreetwo} derives the existence of a fixed point with negative potential from the corresponding local bifurcation condition. The degree-two proof constructs a fixed point with negative potential that satisfies both density-evolution equations, without assuming replica symmetry, and bounds the transmitted SPA density-evolution error after termination.

\subsection{From erasure messages to general symmetric messages}
For Theorem~\ref{thm:main}, the degree-three BEC proof supplies the structure of the argument. That proof first separates trivial and nontrivial fixed points. The trivial values are explicit, while positivity at every nontrivial fixed point requires a separate argument. Two erasure probabilities describe density evolution exactly, and parameterizing the nontrivial fixed points reduces their potential positivity to a polynomial sign problem in one variable. Sturm's theorem handles $3\leq\ell\leq164$, and analytic inequalities handle $\ell\geq165$~\cite{OkazakiKasai2014}. The finite root counts establish a sign on an entire interval.

The BMS proof uses the same classification. The trivial fixed points in Definition~\ref{def:fixedpoint-sets} are $\gd=(\Delta_\infty,\Delta_\infty)$ and $(\Delta_0,c)$, with potential values $0$ and $C(c)-R$, respectively (Section~\ref{sec:potential}, equation~\eqref{eq:trivial-potential}). Thus, when $C(c)>R$, the remaining task is to prove $U(a,b;c)>0$ for every $(a,b)\in\Fnt(c)$. A single erasure probability or reliability statistic does not determine a general BMS message density, so we retain the density-evolution equations~\eqref{eq:DE} of Section~\ref{sec:density-evolution}.

From this infinite-dimensional fixed-point condition, we extract finitely many inequalities that every fixed point must satisfy. Section~\ref{sec:bounds} uses the moments and mutual information of Definition~\ref{def:bms-measures} to write
\[
m=M_1(a),\qquad n=M_1(b),\qquad t=I(b).
\]
Here $m,n$ are the mean squared absolute reliabilities of messages from punctured and transmitted variables, respectively, and $t$ is the mutual information of the transmitted-variable message. Check-convolution moments multiply exactly by~\eqref{eq:moments}, so~\eqref{eq:q} gives $M_1(q_1)=m^{r-1}n^g$. For variable convolution, we apply the comparison inequality of Section~\ref{sec:extremal-replacement} (Lemma~\ref{lem:BSC}). At a fixed point, the updated density equals $a$, so a lower bound on its updated moment cannot exceed $m$. This gives the necessary condition $m\geq\Rk_{\ell-1}(m^{r-1}n^g)$ in~\eqref{eq:cm}. The function $\Rk_k(q)$ is the mean squared reliability obtained by combining $k$ BSC observations with squared reliability $q$, conditionally independent given the same bit, as calculated in~\eqref{eq:Rformula}. Similarly, the transmitted-variable update and the information comparison give~\eqref{eq:ct}, while entropy bounds give~\eqref{eq:ch}. Lemma~\ref{lem:constraints} collects these necessary conditions.

An enclosure in this sense is not an exact scalar parameterization of fixed points. The implication is one-way: if a density fixed point exists, its statistics $(m,n,t)$ must satisfy these inequalities. For example, if $\Rk_{\ell-1}(m^{r-1}n^g)>m$ at a scalar pair $(m,n)$, no fixed point can have those statistics. Likewise, if the lower bound on $t$ in~\eqref{eq:ct} exceeds $n$, it contradicts $t\leq n$ and excludes that pair. Conversely, satisfying all these inequalities need not produce a density pair satisfying both updates in~\eqref{eq:DE}. The remaining scalar region therefore contains the statistics of every actual fixed point, and may also contain points that correspond to none. The BSC comparisons supply bounds; they do not replace the original density evolution by a recursion restricted to BSC messages.

On the remaining region, we use the fixed-point decomposition $U=S+(t-R)D$ in~\eqref{eq:decomp}. The functionals $D$ and $S=S_{r,g}(a,b)$ are defined in~\eqref{eq:D} and~\eqref{eq:S}; Lemma~\ref{lem:reduction} gives $(t-R)D>0$ when $C(c)>R$ and $(a,b)\ne\gd$. Although $S$ depends on all message moments, equations~\eqref{eq:Smoment}--\eqref{eq:Gbound} bound it below using $m,n,t$ by allowing the higher moments to vary over their entire permitted ranges. The bound therefore applies to every density fixed point with the same scalar statistics. Proving $S\geq0$ suffices for $U>0$. Even when the lower bound on $S$ is negative, a positive lower bound after adding $(t-R)D$ still suffices (Section~\ref{sec:certificates}, equation~\eqref{eq:positiveU}).

For the degree-three proof, Lemmas~\ref{lem:weaktrans} and~\ref{lem:weak} treat $n\leq14/25$ and $m\leq4/5$, respectively, by analytic bounds. Proposition~\ref{prop:large} in Section~\ref{sec:large-degrees} handles the remaining region for $\ell\geq33$. For $4\leq\ell\leq32$, the closed rectangle $\mathcal B_\ell$ in~\eqref{eq:root} contains the $(m,n)$ statistics of every fixed point still to be treated. We do not discretize $t$ independently: on each rectangle, we use its lower bound $\tau$ from~\eqref{eq:boxdefs}. The six tests in Section~\ref{sec:certificates} prove that each smaller rectangle either contains no fixed point other than $\gd$, or has $U>0$ at every fixed point it contains. Proposition~\ref{prop:finite} guarantees complete coverage, including shared boundaries, with no unclassified region. Because the bounds also cover scalar points that correspond to no fixed point, locating or counting the actual fixed points is unnecessary. Table~\ref{tab:bec_correspondence} summarizes these steps and their locations in the proof.

\begin{table}[htbp]
\centering\small
\caption{Correspondence between the BEC proof~\cite{OkazakiKasai2014} and the present BMS capacity proof for $r=g=3$. References in the last column point to this paper.}\label{tab:bec_correspondence}
\begin{tabular}{@{}>{\raggedright\arraybackslash}p{.15\linewidth}>{\raggedright\arraybackslash}p{.27\linewidth}>{\raggedright\arraybackslash}p{.50\linewidth}@{}}
\toprule
Proof step & BEC & General BMS channels\\\midrule
Message description & Two erasure probabilities describe the densities exactly. & Retain both densities $a,b$ and both updates: Section~\ref{sec:density-evolution}, equations~\eqref{eq:q}--\eqref{eq:DE}.\\[5pt]
Fixed-point classification & Evaluate the two trivial points; parameterize the nontrivial points. & Use $\Fnt(c)$ (Definition~\ref{def:fixedpoint-sets}); evaluate $U(\gd;c)=0$ and $U(\Delta_0,c;c)=C(c)-R$ (Section~\ref{sec:potential}, \eqref{eq:trivial-potential}).\\[5pt]
Scalar reduction & Reduce nontrivial fixed-point positivity to a one-variable polynomial sign problem. & Use $m=M_1(a)$, $n=M_1(b)$, $t=I(b)$; derive necessary moment, information, and entropy bounds (Section~\ref{sec:extremal-replacement}, Lemma~\ref{lem:constraints}, \eqref{eq:cm}--\eqref{eq:ch}).\\[5pt]
Potential bound & Evaluate the potential on the parameterized fixed points. & At a fixed point, use $U=S+(t-R)D$ (Lemma~\ref{lem:reduction}, \eqref{eq:decomp}); bound all higher moments through \eqref{eq:Smoment}--\eqref{eq:Gbound} in Section~\ref{sec:bounds}.\\[5pt]
Analytic regions & Use the scalar polynomial and analytic inequalities. & Treat $n\leq14/25$ and $m\leq4/5$ (Section~\ref{sec:bounds}, Lemmas~\ref{lem:weaktrans} and~\ref{lem:weak}, \eqref{eq:weakb}, \eqref{eq:weaka}).\\[5pt]
Finite degrees & Sturm root counts certify the interval sign for $3\leq\ell\leq164$. & For $4\leq\ell\leq32$, cover $\mathcal B_\ell$ in \eqref{eq:root}; exclude fixed points or prove $U>0$ on each rectangle (Section~\ref{sec:certificates}, Proposition~\ref{prop:finite}; six tests, including \eqref{eq:positiveU}).\\[5pt]
Large degrees & Analytic estimates for $\ell\geq165$. & Analytically handle the region left by the preceding bounds for $\ell\geq33$ (Section~\ref{sec:large-degrees}, Proposition~\ref{prop:large}, \eqref{eq:tail1}--\eqref{eq:tail2}).\\[5pt]
Positive gap & Combine the trivial and nontrivial fixed-point evaluations. & Separate $\gd$ from $\Fneq(c)$ and use compactness and continuity (Section~\ref{sec:certificates}, Lemma~\ref{lem:separation}, Theorem~\ref{thm:gap}, \eqref{eq:gap}).\\[5pt]
Coupled decoding & Apply threshold saturation for a vector recursion. & Use a one-sided shift of the two-type density recursion and the positive gap (Section~\ref{sec:saturation}, Theorem~\ref{thm:saturation}).\\\bottomrule
\end{tabular}
\end{table}

The density calculus for BMS codes~\cite{KumarYoungMacrisPfister2014} and extremal information-combining inequalities~\cite{LandHuber2006} provide the underlying tools. For degree three, the MN-specific steps are a stationary decomposition of the two-type potential, sufficient bounds on every fixed point other than the perfect-information fixed point, and a density-valued shift argument for the stated termination. A separate projection argument establishes the actual transmitted rate after puncturing. The density-evolution criterion in Corollary~\ref{prop:widthmain} and the operational conclusion in Theorem~\ref{thm:main} distinguish the two steps. For degree two, minimization over one BMS density constructs a fixed point with negative potential, and the finite terminated potential bounds the limiting transmitted error away from zero.

The reusable density identities, scalar necessary conditions, and conditional results retain general $(\ell,r,g)$. Each result states its degree restrictions, and degree-derived coefficients remain symbolic after specialization.

\subsection{Proof outline}\label{sec:outline}
Shortening makes the exterior bits known. Their information can then propagate into the active chain. Whether this propagation completes is controlled here by the signs of the potential at uncoupled fixed points. Theorems~\ref{thm:main} and~\ref{thm:degreetwo} establish opposite signs in their respective degree regimes and derive the corresponding decoding conclusions.

\paragraph*{1. Separate trivial and nontrivial fixed points}
A BEC message is either known or erased. General BMS messages can have many levels of confidence, so density evolution follows their full distributions. A fixed point is a pair of distributions unchanged by an update. As in the BEC proof~\cite{OkazakiKasai2014}, we first distinguish two trivial fixed points: the perfect-information pair $\gd=(\Delta_\infty,\Delta_\infty)$ and the channel-observation pair $(\Delta_0,c)$. Here $\Delta_\infty$ represents perfect information and $\Delta_0$ represents absent information. All other fixed points form $\Fnt(c)$; Definition~\ref{def:fixedpoint-sets} defines both sets.

Direct substitution into the potential in Definition~\ref{def:potential} gives
\[
U(\gd;c)=0,\qquad U(\Delta_0,c;c)=C(c)-R.
\]
For a BEC of erasure probability $\epsilon$, these are the points $(0,0)$ and $(1,\epsilon)$, with potential values $0$ and $1-R-\epsilon$. Thus, when $C(c)>R$, the second trivial fixed point already has positive potential. The unresolved question is whether \emph{every point $x\in\Fnt(c)$} also has positive potential. This is the same division of the positivity argument as in the BEC proof.

\paragraph*{2. Prove positivity at nontrivial fixed points and obtain a positive gap}
For $r=g=3$ and every $\ell\geq4$, we prove this remaining positivity statement on every BMS channel with $C(c)>R$. With $D(a,b)$ and $S_{r,g}(a,b)$ defined in~\eqref{eq:D}--\eqref{eq:S}, the fixed-point reduction~\eqref{eq:decomp} gives
\[
U=S+(I(b)-R)D.
\]
The transmitted-side message $b$ includes the channel observation, so $I(b)\geq C(c)>R$; at $(a,b)\in\Fnt(c)$, $D(a,b)>0$. The second term is therefore strictly positive. We bound $S$ from below and show that the sum is positive, including cases where the bound on $S$ is negative.

In the BEC proof, a one-parameter description of the nontrivial fixed points reduces positivity to a polynomial sign problem~\cite{OkazakiKasai2014}. General BMS fixed points are pairs of densities. Section~\ref{sec:bounds} instead bounds their moments and mutual information by scalar necessary conditions. These conditions may include scalar points that do not represent any density fixed point; proving the inequalities throughout this larger set is sufficient. Convexity, tangents, and chords treat broad regions and all $\ell\geq33$. For $4\leq\ell\leq32$, exact interval certificates cover the remaining regions by closed rectangles, including all rounding and series errors. This completes the positivity proof for every $x\in\Fnt(c)$.

Combining this result with $U(\Delta_0,c;c)=C(c)-R>0$ gives positive potential at \emph{every $x\in\Fneq(c)$}, with $\Fneq(c)$ as in Definition~\ref{def:gap}. For the threshold-saturation argument used here, we establish a strictly positive lower bound on this whole set; pointwise positivity alone would still allow values approaching zero. Lemma~\ref{lem:separation} separates the set from $\gd$. Compactness and continuity then give an attained, strictly positive minimum $\gap(c)$, which includes both the second trivial fixed point and $\Fnt(c)$.

\paragraph*{3. Construct a fixed point with negative potential at degree two}
For $r=g=2$ and $\ell\in\{3,4,5\}$, we instead construct a fixed point with negative potential on some BSCs with $C(c)>R$. Such a point $p$ must belong to $\Fnt(c)$, since the two trivial fixed points have potentials $0$ and $C(c)-R>0$. Thus the positivity required in step 2 fails on these channels. Section~\ref{sec:degreetwo} temporarily fixes the transmitted-side density to $c$ and minimizes $V(a;c)$ from~\eqref{eq:Vdefinition} over the full space of BMS message densities sent by punctured variables. Moving from a minimizer toward its punctured-variable update gives a derivative equal to a negative sum of squared moment differences. A minimum cannot have such a descent direction, so it satisfies the fixed-point equation for the punctured-variable messages.

A mixture perturbation around absent information about the punctured bits decreases $V$ below $1-R$ when
\[
(\ell-1)M_1(c)^g>1.
\]
This is the strict side of the paramagnetic bifurcation condition in the earlier statistical-mechanical analysis~\cite{TanakaSaad2003}. Elementary rational inequalities verify it at the BSC capacity boundary for $\ell\in\{3,4,5\}$. The strict decrease persists on slightly better channels. An entropy comparison gives $U(a,c;c)\leq V(a;c)-H(c)$. Starting the full recursion at this minimizing density yields an iteration monotone in degradation order whose potential cannot increase. Its limit satisfies both density-evolution equations and has negative potential.

\paragraph*{4. Compare the boundary contribution with the interior}
For positive $\gap(c)$, Section~\ref{sec:saturation} compares the terminated chain to a one-sided chain. If the comparison recursion converged to a fixed density profile without perfect information at every position, its endpoint would have potential at least $\gap(c)$. Shifting the profile by one position with fixed exterior densities changes its potential by $\delta\mathcal U$ satisfying
\[
\delta\mathcal U\leq-\gap(c),\qquad
\delta\mathcal U\geq-\frac{\mathcal K_{\ell,r,g}}{2w}.
\]
For sufficiently large $w$, these bounds contradict each other, and the recursion must converge to perfect information.

For a fixed point $p$ with $u=-U(p;c)>0$, initialize the check-input density pairs to be updated at $p$. The finite terminated potential is
\[
-Lu+(w-1)Q(p),\qquad 0\leq Q(p)\leq1.
\]
Updates monotone in degradation order can only decrease it. When $L/w\to\infty$, the boundary contribution cannot compensate for the negative interior contribution. Comparison with the initialization from channel observations and a bound relating the two message types then give a strictly positive transmitted-bit error. This bounds the error of the transmitted bits themselves.

\paragraph*{5. Establish the actual transmitted rate}
For degree three, Section~\ref{sec:rate} transfers density evolution to finite code realizations and establishes their actual transmitted rate. Reliability of the punctured bits controls the dimension lost under projection; auxiliary BECs give a matching upper bound. For degree two, the negative-potential argument directly gives the density-evolution obstruction in Theorem~\ref{thm:degreetwo}.

\section{Preliminaries on BMS densities}\label{sec:preliminaries}
The fixed-point strategy requires a description of messages that retains their full BMS distributions. The information measures and convolution identities below give the values at the two trivial fixed points and the estimates used to control the nontrivial ones. The degradation order also provides the comparisons needed to pass from uncoupled fixed points to terminated decoding.

\subsection{Message densities and information measures}\label{sec:message-measures}
\begin{definition}[BMS densities and information measures]\label{def:bms-measures}
A symmetric log-likelihood-ratio density (L-density) can be represented as a mixture of binary symmetric channel (BSC) observations whose absolute reliability $D\in[0,1]$ is revealed to the receiver~\cite[Section II-A and Appendix I]{KumarYoungMacrisPfister2014}. Define
\[
H(x)=\E h_2((1-D)/2),\quad I(x)=1-H(x),\quad
Z(x)=\E\sqrt{1-D^2},\quad M_k(x)=\E D^{2k}.
\]
The decision from the sign of a log-likelihood ratio with density $x$, with an unbiased random decision at zero, has bit error probability obtained by averaging the crossover probability of the revealed BSC~\cite[Section II-A and Appendix I]{KumarYoungMacrisPfister2014}:
\begin{equation}\label{eq:Pe}
P_e(x)=\E[(1-D)/2].
\end{equation}
Here $h_2$ is the binary entropy function, with logarithms to base two. The transmission channel has density $c$ and capacity $C(c)=I(c)$. Let $\Delta_\infty$ and $\Delta_0$ denote perfect and absent information.
\end{definition}

Reliability zero means that an observation gives no information about the bit, whereas reliability one means that it determines the bit. A BSC has a constant absolute reliability; a BEC has reliability zero or one. A general BMS density allows any probability law on this interval. This representation lets us compare messages through moments while retaining their full distributions.

\subsection{Convolutions and degradation order}
\begin{definition}[Convolutions and degradation order]\label{def:convolution-order}
Variable convolution $\star$ sums independent log-likelihood ratios (LLRs) of observations of the same bit. Check convolution $\cc$ multiplies their signed hyperbolic-tangent reliabilities. All observation independence is conditional on the bit. Write $x\preceq y$ when $y$ is physically degraded from $x$, and use componentwise order for pairs. This is the degradation order used in~\cite[Definition 2 and Proposition 3]{KumarYoungMacrisPfister2014}; both convolutions preserve it.
\end{definition}

Thus $x\preceq y$ means that $x$ is at least as informative as $y$: $H(x)\leq H(y)$ and $I(x)\geq I(y)$. An update that improves information moves downward in this order. Entropy alone does not determine degradation order; the comparisons below use the ordering of the full observations.

The endpoint rules used below follow from these observation operations~\cite[Section II-A]{KumarYoungMacrisPfister2014}:
\begin{equation}\label{eq:endpoint-convolutions}
\begin{aligned}
x\star\Delta_0&=x,&x\star\Delta_\infty&=\Delta_\infty,\\
x\cc\Delta_0&=\Delta_0,&x\cc\Delta_\infty&=x,\\
H(\Delta_0)&=1,&H(\Delta_\infty)&=0.
\end{aligned}
\end{equation}
An absent observation has LLR zero, so adding it leaves a variable message unchanged. A perfect observation determines the bit. At a check, multiplying by reliability zero erases the outgoing information, whereas multiplying by reliability one leaves the other input unchanged. Finally, the entropy values are $h_2(1/2)=1$ and $h_2(0)=0$.

\subsection{Entropy identities and basic inequalities}
For probability densities $x,y$, entropy duality~\cite[Proposition 4]{KumarYoungMacrisPfister2014}, the entropy series~\cite[Proposition 7]{KumarYoungMacrisPfister2014}, and the check-moment product rule~\cite[Proposition 6(iii)]{KumarYoungMacrisPfister2014} give, respectively,
\begin{align}
H(x\star y)+H(x\cc y)&=H(x)+H(y),\label{eq:duality}\\
I(x)&=\sum_{k\geq1}\gamma_kM_k(x),&
\gamma_k&=\frac1{\ln2\,2k(2k-1)},\quad\sum_k\gamma_k=1,\label{eq:moments}\\
M_k(x\cc y)&=M_k(x)M_k(y).\nonumber
\end{align}
Differentiation of a density means differentiation along a mixture. For $x_s=(1-s)x+sy$, $0\leq s\leq1$, the derivative is $y-x$, a signed measure of total mass zero. The notation $H(y-x)$ means $H(y)-H(x)$, the linear extension of the entropy kernel to this difference. Its sign is governed by the information comparison between $x$ and $y$.

When signed measures occur below, $H$ is extended linearly. The signed duality identities and entropy signs used here follow from~\cite[Propositions 5 and 8(iii)]{KumarYoungMacrisPfister2014}. For a mass-zero signed measure $\nu$, duality reads $H(x\star\nu)+H(x\cc\nu)=H(\nu)$. For two such signed measures the sum is zero. If both signed measures are degraded-minus-informative differences, their check entropy is nonpositive by~\eqref{eq:moments}, and their variable entropy is nonnegative. Opposite orientations reverse these signs.

The elementary entropy and Bhattacharyya bounds needed below are
\begin{equation}\label{eq:HZ}
1-M_1(x)\leq H(x)\leq Z(x)\leq\sqrt{1-M_1(x)},\quad
Z(x\star y)=Z(x)Z(y),\quad Z(x\cc y)\leq Z(x)+Z(y).
\end{equation}
The bound of entropy by the Bhattacharyya parameter and the variable-convolution identity are standard~\cite[Section II-D and Appendix II-D]{KumarYoungMacrisPfister2014}. The lower bound $1-M_1(x)\leq H(x)$ follows from~\eqref{eq:moments} by $M_k(x)\leq M_1(x)$ and $\sum_k\gamma_k=1$; the upper bound on $Z$ is Jensen's inequality~\cite[Section 3.1.8]{BoydVandenberghe2004}. Conditional on squared reliabilities $u,v$, the check Bhattacharyya kernel is $\sqrt{1-uv}\leq\sqrt{1-u}+\sqrt{1-v}$. Also,
\begin{equation}\label{eq:productH}
H(x\star y)\geq H(x)H(y).
\end{equation}
Indeed, both moment sequences decrease with $k$, so their covariance under weights $\gamma_k$ is nonnegative. Applying the moment product rule and~\eqref{eq:moments} to this covariance gives $I(x\cc y)\geq I(x)I(y)$~\cite[Propositions 6(iii) and 7]{KumarYoungMacrisPfister2014}. Using duality~\eqref{eq:duality}, this lower bound enters with a positive sign after writing $H(x\cc y)=1-I(x\cc y)$:
\[
\begin{aligned}
H(x\star y)
 &=H(x)+H(y)-1+I(x\cc y)\\
 &\geq H(x)+H(y)-1+(1-H(x))(1-H(y))\\
 &=H(x)H(y).
\end{aligned}
\]

\section{MN ensemble and density evolution}\label{sec:model}
We use the conventional $(\ell,r,g)$ parametrization of MN codes~\cite{KasaiSakaniwa2011,OkazakiKasai2014}, with integers $\ell>r\geq2$ and $g\geq2$. The design rate is $R=r/\ell$. The punctured-variable degree is $\ell$; each check has $r$ sockets for punctured variables and $g$ for transmitted variables; the transmitted-variable degree is also $g$. Table~\ref{tab:degrees_model} records the convention. Keeping these parameters explicit allows each subsequent result to state its degree restrictions where they are needed. The construction and density evolution below specify the fixed points to be classified as trivial or nontrivial. They also specify the shortening and initialization needed to translate the sign of the uncoupled potential into a statement about terminated decoding.

\subsection{Terminated ensemble}
Let $L\geq1$ be the active chain length, $w\geq1$ the smoothing width, and $M$ the section-size parameter. At each section there are $rM/\ell$ punctured variables, $M$ transmitted variables, and $M$ checks. Choose $M$ so that $rM/\ell$, $rM/w$, and $gM/w$ are integral. Start with a periodic chain longer than $L+2w$. For each section, independently partition the $rM$ sockets at punctured variables into $w$ groups of size $rM/w$, and the $gM$ transmitted-variable sockets into $w$ groups of size $gM/w$, labelled by offsets $0,\ldots,w-1$. Partition the corresponding check sockets independently into incoming-offset groups of the same sizes, and uniformly match corresponding groups. A variable socket at $i$ with offset $k$ reaches check section $i+k$.

\begin{table}[htbp]
\centering
\caption{Degree convention for $(\ell,r,g)$-MN codes.}\label{tab:degrees_model}
\begin{tabular}{lccc}
\toprule
Variable type & Variables per section & Variable degree & Sockets per check\\\midrule
Punctured & $rM/\ell$ & $\ell$ & $r$\\
Transmitted & $M$ & $g$ & $g$\\\bottomrule
\end{tabular}
\end{table}

Keep variable sections $1,\ldots,L$ active, and shorten both types of all other variables to zero. Choose the periodic indexing so that no wraparound touches the active interval. Only checks at $1,\ldots,L+w-1$ can impose nontrivial constraints on the active variables. Only the $N_{\mathrm{tx}}=LM$ active transmitted bits use the transmission channel. For fixed $L,w$ and a fixed finite neighborhood, balanced random socket partitions and matchings have asymptotically independent uniform offsets as $M\to\infty$. This local weak limit gives the recursion below. The construction does not prescribe offsets separately at each variable node.

\subsection{Density evolution and bit error}\label{sec:density-evolution}
The perfect-information pair is $\gd=(\Delta_\infty,\Delta_\infty)$.

For a pair $x=(a,b)$, apply the sum-product convolution rules~\cite[Section II-A]{KumarYoungMacrisPfister2014} to the socket counts in Table~\ref{tab:degrees_model}, and set
\begin{align}
q_1(x)&=a^{\cc(r-1)}\cc b^{\cc g},\qquad q_2(x)=a^{\cc r}\cc b^{\cc(g-1)},\label{eq:q}\\
s(x)&=a^{\cc r}\cc b^{\cc g},\qquad \mathsf g(x)=(q_1(x),q_2(x)),\nonumber\\
\mathsf f_c(y)&=(y_1^{\star(\ell-1)},c\star y_2^{\star(g-1)}),\qquad\DE=\mathsf f_c\circ\mathsf g.\label{eq:DE}
\end{align}
A message excludes its destination socket; hence the exponents in~\eqref{eq:q}--\eqref{eq:DE}. The bit posterior densities are $q_1^{\star\ell}$ and $c\star q_2^{\star g}$.

For every BMS channel density $c$, the two trivial fixed points can be checked explicitly. At $(a,b)=(\Delta_0,c)$, each check output in~\eqref{eq:q} contains at least one factor $\Delta_0$: the punctured exponents are $r-1\geq1$ in $q_1$ and $r\geq2$ in $q_2,s$. Hence $q_1=q_2=s=\Delta_0$ by~\eqref{eq:endpoint-convolutions}. At $\gd$, all check inputs and outputs are $\Delta_\infty$. Equation~\eqref{eq:DE} therefore gives
\[
\begin{aligned}
\DE(\Delta_0,c)
 &=\bigl(\Delta_0^{\star(\ell-1)},c\star\Delta_0^{\star(g-1)}\bigr)
 =(\Delta_0,c),\\
\DE(\gd)
 &=\bigl(\Delta_\infty^{\star(\ell-1)},c\star\Delta_\infty^{\star(g-1)}\bigr)
 =(\Delta_\infty,\Delta_\infty)=\gd.
\end{aligned}
\]
The exponents $\ell-1$ and $g-1$ are positive under the stated degree assumptions, so the endpoint rules apply to both updates.
\begin{definition}[Trivial and nontrivial fixed points]\label{def:fixedpoint-sets}
Following the BEC convention of~\cite{OkazakiKasai2014}, we call these two points the \emph{trivial fixed points} and all other fixed points \emph{nontrivial fixed points}. More precisely, define
\begin{equation}\label{eq:fixedpoint-sets}
\mathcal F_{\mathrm{triv}}(c)=\{\gd,(\Delta_0,c)\},\qquad
\Fnt(c)=\{x:\DE(x)=x\}\setminus\mathcal F_{\mathrm{triv}}(c).
\end{equation}
\end{definition}

At $(\Delta_0,c)$, the punctured-variable messages contain no information and the transmitted-variable messages contain only the channel observation. The uncoupled recursion initialized with these observations remains at this fixed point.

For $c=\mathrm{BEC}(\epsilon)$, the two points correspond in erasure-probability coordinates to $(x_1,x_2;\epsilon)=(0,0;\epsilon)$ and $(1,\epsilon;\epsilon)$. As $\epsilon$ varies, these describe the two lines of trivial fixed points in the BEC analysis~\cite{OkazakiKasai2014}. At a fixed $\epsilon$, they specify two points. The BMS counterparts are the families $(\gd;c)$ and $(\Delta_0,c;c)$ as the channel density varies.

Initialize messages using the channel observations: set $(a_{i,0},b_{i,0})=(\Delta_0,c)$ on $1\leq i\leq L$ and use $\gd$ outside. A density profile is a sequence of density pairs indexed by spatial position. For variable-section densities $a_{i,t},b_{i,t}$, define the check-input and variable-input pairs by
\[
x_{j,t}=\frac1w\sum_{k=0}^{w-1}(a_{j-k,t},b_{j-k,t}),\qquad
y_{i,t}=\frac1w\sum_{k=0}^{w-1}\mathsf g(x_{i+k,t}).
\]
For $1\leq i\leq L$, the next pair is $\mathsf f_c(y_{i,t})$; otherwise it is clamped to $\gd$ at every iteration. These averages are mixtures of measures.

For example, a check input first selects one of the $w$ neighboring variable sections uniformly and then has that section's message law. The average describes this random selection. The convolution operations subsequently combine independent incoming messages. This distinction between selecting a message and combining messages accounts for the order of averaging and convolution in the recursion. Equivalently, putting $h_{i,t}=\mathsf f_c(y_{i,t})$ on active sections and $h_{i,t}=\gd$ outside,
\begin{equation}\label{eq:coupled}
x_{j,t+1}=\frac1w\sum_{k=0}^{w-1}h_{j-k,t}.
\end{equation}
The coupled posteriors use powers $\ell$ and $g$ of the two components of $y_{i,t}$, with channel $c$ in the transmitted component. Both exterior shortenings are part of~\eqref{eq:coupled}.

Apply~\eqref{eq:Pe} to the transmitted posterior associated with $y_{i,t}$ under this initialization, and define
\begin{equation}\label{eq:txDEerror}
\begin{aligned}
P_{\mathrm{tx}}^{(t)}(L,w;c)
&=\frac1L\sum_{i=1}^{L}P_e(c\star y_{i,t,2}^{\star g}),\\
\overline P_{\mathrm{tx}}(L,w;c)
&=\lim_{t\to\infty}P_{\mathrm{tx}}^{(t)}(L,w;c).
\end{aligned}
\end{equation}
Here $t=0$ is the first check-node update. For fixed $L,w$, the errors are nonincreasing because this recursion becomes more informative at each update by order preservation~\cite[Proposition 3]{KumarYoungMacrisPfister2014}, so the limit exists. Each finite-iteration density-evolution value is obtained at fixed $t$ by taking $M\to\infty$, before this iteration limit.

\section{Potential and fixed-point reduction}\label{sec:potential}
The two trivial fixed points identified in Section~\ref{sec:model} admit explicit potential values. When $C(c)>R$, their evaluation leaves positivity at the points $x\in\Fnt(c)$ defined in~\eqref{eq:fixedpoint-sets} as the remaining sign question. We define the potential, distinguish this nontrivial set from the larger set used for the gap, and derive the fixed-point decomposition used in the subsequent sign estimates. These identities use the general degree parameters of Section~\ref{sec:model} and the entropy identities of Section~\ref{sec:preliminaries}.

\begin{definition}[Uncoupled potential]\label{def:potential}
We use the calculus for differentiation along mixtures~\cite[Propositions 14--15]{KumarYoungMacrisPfister2014}. For the two-type MN recursion, define the potential per transmitted bit by
\begin{equation}\label{eq:U}
\begin{aligned}
U(a,b;c)={}&-(r+g-1)H(s)+rH(q_1)+gH(q_2)\\
&-R H(q_1^{\star\ell})-H(c\star q_2^{\star g}).
\end{aligned}
\end{equation}
\end{definition}

At $(a,b)=(\Delta_0,c)$, the check-output calculation in Section~\ref{sec:density-evolution} and the endpoint rules~\eqref{eq:endpoint-convolutions} give
\[
q_1=q_2=s=\Delta_0,\qquad
q_1^{\star\ell}=\Delta_0,\qquad
c\star q_2^{\star g}=c.
\]
Thus the first four entropies in~\eqref{eq:U} equal one, whereas the last is $H(c)$. Substituting them separately yields
\begin{equation}\label{eq:trivial-potential}
\begin{aligned}
U(\Delta_0,c;c)
 &=-(r+g-1)\cdot1+r\cdot1+g\cdot1-R\cdot1-H(c)\\
 &=\bigl[-(r+g-1)+r+g\bigr]-R-H(c)\\
 &=1-R-H(c)=C(c)-R.
\end{aligned}
\end{equation}
The last equality uses $C(c)=1-H(c)$ from Section~\ref{sec:message-measures}. The two occurrences of $c$ in $U(\Delta_0,c;c)$ have different roles: the second argument is the transmitted-variable message density $b=c$, and the argument after the semicolon is the transmission channel. The coefficients involving $r,g$ cancel to one; this calculation does not require $r=g=3$.

At $\gd$, all three check-output densities are $\Delta_\infty$, as are both posterior densities, since $c\star\Delta_\infty=\Delta_\infty$. Every entropy in~\eqref{eq:U} is therefore zero:
\[
U(\gd;c)=-(r+g-1)\cdot0+r\cdot0+g\cdot0-R\cdot0-0=0.
\]

This potential also has an exact relation to the statistical-mechanical expression with parameters $(K,C,L)=(r,\ell,g)$~\cite{TanakaSaad2003}. Let $\mathcal F_{\mathrm{RS}}$ be the unextremized replica-symmetric free-energy expression per information bit, with natural logarithms, at the matched temperature and with unbiased information bits. The information bits in that expression correspond to the punctured bits in the present construction. Restrict the variable-message laws to BMS densities $a,b$ and set the check-message laws to $q_1,q_2$. Let $f_{\mathrm{ferro}}(c)$ be the value of this expression at the perfect-information point $\gd$, called the ferromagnetic solution in that analysis. Its entropy representation gives
\begin{equation}\label{eq:rsU}
U(a,b;c)=\frac{R}{\ln2}
\bigl(\mathcal F_{\mathrm{RS}}(a,b;c)-f_{\mathrm{ferro}}(c)\bigr).
\end{equation}
Here $R$ converts the normalization from punctured to transmitted bits, and $1/\ln2$ converts natural-log units to bits. Subtract the channel reference before integration if separate logarithmic expectations are not finite; they are finite for the BSCs used below. This algebraic identity requires check consistency but no variable fixed-point equations. It does not identify the replica-symmetric expression with the true MAP conditional entropy. The verification is given at the end of this section.

\begin{definition}[Fixed-point gap]\label{def:gap}
Define the set of fixed points other than the perfect-information point and its potential infimum by
\begin{equation}\label{eq:gapdefinition}
\Fneq(c)=\{x:\DE(x)=x,\ x\neq\gd\},\qquad
\gap(c)=\inf_{x\in\Fneq(c)}U(x;c).
\end{equation}
\end{definition}

The set includes the other trivial fixed point $(\Delta_0,c)$ and is therefore nonempty. This convention agrees with the fixed-point set in the BEC potential-threshold definition of~\cite{OkazakiKasai2014}, which also excludes only the perfect-information point. In terms of the nontrivial fixed-point set $\Fnt(c)$ in~\eqref{eq:fixedpoint-sets},
\[
\begin{aligned}
\Fneq(c)&=\{(\Delta_0,c)\}\cup\Fnt(c),\\
\gap(c)&=\min\left\{C(c)-R,\inf_{x\in\Fnt(c)}U(x;c)\right\},
\end{aligned}
\]
where the infimum of the empty set is $+\infty$. When $C(c)>R$, the other trivial fixed point already has positive potential, so the remaining positivity argument concerns $x\in\Fnt(c)$. The gap nevertheless takes both parts of $\Fneq(c)$ into account. For general degrees, neither positivity nor attainment is assumed here. This fixed-point infimum differs from the energy gap defined outside the attraction basin in~\cite[Definition 25(ii)]{KumarYoungMacrisPfister2014}. Lemma~\ref{lem:descent} bounds the potential at a non-perfect right endpoint of the limiting comparison profile from below by the fixed-point gap, which suffices for the shift argument.

The connection to the BEC potential is exact. Let $a=\mathrm{BEC}(x_1)$,
$b=\mathrm{BEC}(x_2)$, and $c=\mathrm{BEC}(\epsilon)$, and denote the check-output erasure probabilities by
\begin{align*}
\eta_1&=1-(1-x_1)^{r-1}(1-x_2)^g,&
\eta_2&=1-(1-x_1)^r(1-x_2)^{g-1},\\
\eta_s&=1-(1-x_1)^r(1-x_2)^g.
\end{align*}
BEC entropy equals erasure probability, and variable convolution multiplies erasure probabilities, as in the BEC density-evolution calculation of~\cite{OkazakiKasai2014}. Thus~\eqref{eq:DE} updates the pair to
$(\eta_1^{\ell-1},\epsilon\eta_2^{g-1})$, and~\eqref{eq:U} becomes
\[
U=-(r+g-1)\eta_s+r\eta_1+g\eta_2-R\eta_1^\ell-\epsilon\eta_2^g.
\]
This is the scalar MN potential obtained from the BEC recursion and potential definition in~\cite{OkazakiKasai2014}. In particular, at the trivial fixed point $(1,\epsilon)$ it gives $U(\Delta_0,c;c)=1-R-\epsilon$, the BEC specialization of $C(c)-R$. Indeed, $x_1=1$ makes $\eta_1=\eta_2=\eta_s=1$, because the powers of $1-x_1$ have positive exponents. Hence the five scalar terms give $-(r+g-1)+r+g-R-\epsilon=1-R-\epsilon$. At $(x_1,x_2)=(0,0)$ all three $\eta$ values are zero, and the potential is zero. The fixed-point identity below provides a way to estimate the same density functional when the messages have general BMS laws.

The reduction has two uses. The first-variation formula expresses the derivative of $U$ through the difference between a message pair and its update, so every fixed point is stationary. At a fixed point other than $\gd$, the decomposition isolates the strictly positive term $(I(b)-R)D$ when $C(c)>R$. It is enough to bound the remaining term $S$ from below; even a negative lower bound can suffice if the positive term compensates for it.

\begin{lemma}[First variation and decomposition at a fixed point]\label{lem:reduction}
Let $(a',b')=\DE(a,b)$. Along any admissible mixture variation,
\begin{equation}\label{eq:dU}
dU=rH((a-a')\star dq_1)+gH((b-b')\star dq_2).
\end{equation}
For arbitrary pairs define
\begin{align}
D(a,b)&=H(a\star q_1),\label{eq:D}\\
S_{r,g}(a,b)&=I(a)I(b)-(r-I(b))I(q_1)\nonumber\\
&\quad -(g-1)I(q_2)+(r+g-2-I(b))I(s).\label{eq:S}
\end{align}
For fixed $r,g$, abbreviate $S=S_{r,g}$.
At a fixed point,
\begin{equation}\label{eq:decomp}
U=S+(I(b)-R)D.
\end{equation}
If $C(c)>R$ and the fixed point differs from $\gd$, then $I(b)>R$ and $D>0$.
\end{lemma}
\begin{proof}
Using $s=a\cc q_1=b\cc q_2$ from~\eqref{eq:q}, rewrite the first three terms of~\eqref{eq:U} as
$H(s)+r[H(q_1)-H(a\cc q_1)]+g[H(q_2)-H(b\cc q_2)]$.
Direct variations in $a,b$ cancel the corresponding terms from $H(s)$.
Signed entropy duality~\cite[Proposition 5]{KumarYoungMacrisPfister2014}, applied to the remaining variations, gives~\eqref{eq:dU}.

In detail, the convolution differentiation rule~\cite[Propositions 14--15]{KumarYoungMacrisPfister2014} selects one of the $\ell$ factors of the punctured posterior entropy and gives $R\ell H(a'\star dq_1)=rH(a'\star dq_1)$. The transmitted posterior similarly gives $gH(b'\star dq_2)$. After the preceding cancellation, the other terms are $rH(a\star dq_1)+gH(b\star dq_2)$; subtraction leaves precisely the two update residuals in~\eqref{eq:dU}.
At a fixed point, substitute~\eqref{eq:DE} into the posterior densities and apply entropy duality~\eqref{eq:duality}, as given in~\cite[Proposition 4]{KumarYoungMacrisPfister2014}, to obtain
\[
\begin{aligned}
H(q_1^{\star\ell})
 &=H(q_1^{\star(\ell-1)}\star q_1)=H(a\star q_1)\\
 &=H(a)+H(q_1)-H(a\cc q_1)=H(a)+H(q_1)-H(s),\\
H(c\star q_2^{\star g})
 &=H((c\star q_2^{\star(g-1)})\star q_2)=H(b\star q_2)\\
 &=H(b)+H(q_2)-H(b\cc q_2)=H(b)+H(q_2)-H(s).
\end{aligned}
\]
Substitution of these two identities into~\eqref{eq:U} and collection of its entropy coefficients yield
\[
V_R=(R+2-r-g)H(s)+(r-R)H(q_1)+(g-1)H(q_2)-RH(a)-H(b).
\]
To verify the remaining algebra, put $t=I(b)=1-H(b)$ in~\eqref{eq:S} and expand each other information term as $I=1-H$:
\[
\begin{aligned}
S={}&t(1-H(a))-(r-t)(1-H(q_1))\\
 &-(g-1)(1-H(q_2))+(r+g-2-t)(1-H(s))\\
={}&-H(b)-tH(a)+(r-t)H(q_1)+(g-1)H(q_2)\\
 &-(r+g-2-t)H(s).
\end{aligned}
\]
Here the constant terms add to $t-1=-H(b)$. Entropy duality~\eqref{eq:duality} gives $D=H(a)+H(q_1)-H(s)$, so adding $(t-R)D$ and collecting each coefficient gives
\[
\begin{aligned}
S+(t-R)D
={}&[-t+(t-R)]H(a)+[(r-t)+(t-R)]H(q_1)\\
 &+(g-1)H(q_2)-H(b)\\
 &+[-(r+g-2-t)-(t-R)]H(s)\\
={}&-RH(a)+(r-R)H(q_1)+(g-1)H(q_2)-H(b)\\
 &+(R+2-r-g)H(s)=V_R.
\end{aligned}
\]
This proves~\eqref{eq:decomp}, because the preceding posterior substitutions established $U=V_R$ at a fixed point.
The fixed-point equation~\eqref{eq:DE} gives $b=c\star q_2^{\star(g-1)}\preceq c$, and entropy monotonicity under degradation~\cite[Section II-C]{KumarYoungMacrisPfister2014} gives $I(b)\geq C(c)>R$.
By~\eqref{eq:productH}, $D=H(q_1^{\star\ell})\geq H(q_1)^\ell$.
If $D=0$, then $q_1$ is perfect. The check-moment product rule and the characterization of perfect information~\cite[Proposition 6(iii)--(iv)]{KumarYoungMacrisPfister2014}, applied to~\eqref{eq:q}, then require both $a,b$ to be perfect. This includes distributions with perfect-information atoms.
\end{proof}

\paragraph*{Relation to the replica-symmetric free energy}
To verify~\eqref{eq:rsU}, start from the MN replica-symmetric free-energy expression~\cite{TanakaSaad2003}. Rewrite its logarithmic expectations using the BMS entropy representation~\cite[Section II-C and Proposition 7]{KumarYoungMacrisPfister2014}. Check consistency~\eqref{eq:q} makes the edge/check terms multiples of $I(s)$. In bits, the punctured-variable term is $H(q_1^{\star\ell})+\ell I(q_1)$, and the channel-node term after reference subtraction is $H(c\star q_2^{\star g})+gI(q_2)$. Their normalized combination is
\[
1+(r+g-1)I(s)-rI(q_1)-gI(q_2)
-R H(q_1^{\star\ell})-H(c\star q_2^{\star g}),
\]
which equals~\eqref{eq:U} after replacing every $I$ by $1-H$ and collecting constants.

\section{Information combining and analytic positivity}\label{sec:bounds}
The two trivial fixed points have already been evaluated in Section~\ref{sec:potential}. For $r=g=3$ and $C(c)>R$, the remaining sign problem is to prove $U>0$ at every $x\in\Fnt(c)$, with $\Fnt(c)$ as in~\eqref{eq:fixedpoint-sets}. We derive scalar necessary conditions for these points and lower bounds for $S_{r,g}(a,b)$ defined in~\eqref{eq:S} and used in~\eqref{eq:decomp}, in place of the exact fixed-point parametrization available for the BEC~\cite{OkazakiKasai2014}. The estimates retain general degree parameters whenever possible. Their specialization proves positivity for all $\ell\geq33$ and reduces $4\leq\ell\leq32$ to the finite sign conditions completed in Section~\ref{sec:certificates}.

Write $A_k=M_k(a)$, $B_k=M_k(b)$, $m=A_1$, $n=B_1$, and $t=I(b)$. These are the moments and information of Section~\ref{sec:message-measures}; the coefficients $\gamma_k$ are defined in~\eqref{eq:moments}. Here $n$ is a local scalar statistic, distinct from block length. Define
\begin{equation}\label{eq:scalar-information}
J(d)=1-h_2((1-d)/2),\qquad f(u)=J(\sqrt u)=\sum_{k\geq1}\gamma_ku^k.
\end{equation}
The series is~\eqref{eq:moments} specialized to a BSC~\cite[Proposition 7]{KumarYoungMacrisPfister2014}. Its positive coefficients establish convexity and monotonicity. Applying Jensen's inequality to $D^2$~\cite[Section 3.1.8, Eq. (3.5)]{BoydVandenberghe2004} gives the lower bound below; $M_k(x)\leq M_1(x)$ and $\sum_k\gamma_k=1$ give the upper bound:
\begin{equation}\label{eq:fbound}
f(M_1(x))\leq I(x)\leq M_1(x).
\end{equation}

\subsection{Extremal replacement of observations}\label{sec:extremal-replacement}

The scalars $m,n,t$ constrain the full density equations. Check convolution is particularly convenient because its moments multiply exactly. For variable convolution, the BSC replacement lemma supplies lower bounds. At a fixed point, the outgoing density must reproduce its own statistics, which turns these bounds into restrictions on where a fixed point can occur.
Comparing BMS observations at fixed mutual information is an information-combining method~\cite{LandHuber2006}. The following proof establishes the BSC comparisons needed here, including the separate comparison at fixed mean squared reliability.
The two replacements below preserve different quantities and are applied separately. In the moment comparison, convexity says that averaging the squared reliability before combining observations can only lower the output moment. In the information comparison, concavity first gives an upper bound for check-combined information; entropy duality converts it into the required lower bound for variable-combined information.

\begin{lemma}[BSC replacement]\label{lem:BSC}
For finitely many conditionally independent BMS observations of one bit, replacing each observation by a BSC with the same mean squared reliability cannot increase the combined mean squared reliability. Replacing each by a BSC with the same mutual information cannot increase the combined mutual information.
\end{lemma}
\begin{proof}
For two observations with squared absolute reliabilities $u,v$, the variable-convolution rule~\cite[Section II-A]{KumarYoungMacrisPfister2014} sums their LLRs. Averaging the resulting squared reliability over the observation signs gives
\[
\psi(u,v)=\frac{u+v-2uv}{1-uv}=1-\frac{(1-u)(1-v)}{1-uv}.
\]
This follows by weighting equal and opposite signs by $(1+\sqrt{uv})/2$ and $(1-\sqrt{uv})/2$. Set $\psi(1,1)=1$ by continuity. For $uv<1$,
\[
\partial_u^2\psi(u,v)=\frac{2v(1-v)^2}{(1-uv)^3}\geq0.
\]
Applying Jensen's inequality~\cite[Section 3.1.8]{BoydVandenberghe2004} separately to each variable shows that replacing either random squared reliability by its mean cannot increase the output moment.

For the information comparison fix the other reliability $v\in[0,1]$ and put $F_v(s)=J(vJ^{-1}(s))$. Differentiating the definition of $J$ and applying the chain rule gives, for $d=J^{-1}(s)\in(0,1)$,
\[
F_v'(s)=v\frac{\operatorname{atanh}(vd)}{\operatorname{atanh}(d)}.
\]
This decreases with $d$ when $0<v<1$. To see this, use the logarithm series~\cite[Eq. (4.6.4)]{NISTDLMF} for twice the inverse hyperbolic tangent and cancel $d$ in the two positive series. Their coefficient ratios are $v^{2j+1}$. In the cross-multiplied derivative, pairing terms indexed by $i>j$ gives a positive multiplier times $(i-j)(v^{2i+1}-v^{2j+1})\leq0$. Thus $F_v$ is concave; $v=0,1$ are immediate. Jensen's inequality~\cite[Section 3.1.8]{BoydVandenberghe2004} shows that a BSC replacement at fixed information increases check-combined information. Rewriting entropy duality~\eqref{eq:duality}~\cite[Proposition 4]{KumarYoungMacrisPfister2014} as $I(x\star y)=I(x)+I(y)-I(x\cc y)$ reverses this comparison for variable combining.

For either comparison, group all observations except the one being replaced into their combined BMS observation. This group remains conditionally independent of the selected observation given the bit. Replace observations successively. Continuity handles the endpoint reliabilities. Increasing the reliability of any replacement BSC increases combined information and moment by physical degradation~\cite[Propositions 3, 6(ii), and 8(iv)]{KumarYoungMacrisPfister2014}.
\end{proof}

For $k\geq1$ let $\Rk_k(q)$ be the moment of $k$ BSC observations with squared reliability $q$. Let $E_d(q)=\sum_{i=0}^{\lfloor d/2\rfloor}\binom d{2i}q^i$. Using the LLR addition rule in Section~\ref{sec:preliminaries} and pairing binomial patterns with $j$ and $k-j$ negative signs gives
\begin{align}\label{eq:Rformula}
1-\Rk_k(q)={}&\sum_{j=0}^{\lfloor(k-1)/2\rfloor}
\frac{\binom kj}{2^{k-1}}\frac{(1-q)^{k-j}}{E_{k-2j}(q)}
+\mathbf1_{\{k\text{ even}\}}\frac{\binom k{k/2}}{2^k}(1-q)^{k/2}.
\end{align}
Indeed, with $d=\sqrt q$, the paired likelihood numerators are $(1+d)^{k-j}(1-d)^j$ and $(1-d)^{k-j}(1+d)^j$. Their contribution to the posterior variance cancels odd powers of $d$ and yields~\eqref{eq:Rformula}. In the even central term, reliability is zero. Denominators satisfy $E_d(q)\geq1$, so the formula includes $q=0,1$.
Choose a rational $d_0\in[0,1)$ with $J(d_0)<R$. For $k\geq1$, let $\mathcal T_{k,d_0}(q)$ be the mutual information obtained by combining one BSC observation of reliability $d_0$ with $k$ BSC observations of reliability $\sqrt q$, all $k+1$ observations being conditionally independent given the same bit. Both $\Rk_k$ and $\mathcal T_{k,d_0}$ increase with $q$ by degradation.

When $g=3$, with $r$ unrestricted, the $g-1$ check observations form a pair and have the explicit information formula
\begin{equation}\label{eq:Tformula}
\begin{aligned}
\mathcal T_{g-1,d_0}(q)={}&J(d_0)+(g-1)f(q)-f(q^{g-1})\\
&-\frac{1+q}{2}f\left(\frac{(g-1)^2d_0^2q}{(1+q)^2}\right).
\end{aligned}
\end{equation}
Here the $g-1$ observations form a pair: their combined reliability is $(g-1)\sqrt q/(1+q)$ with probability $(1+q)/2$ and zero otherwise. Applying entropy duality~\eqref{eq:duality}~\cite[Proposition 4]{KumarYoungMacrisPfister2014} to this pair and then to its combination with the channel gives~\eqref{eq:Tformula}. The factors $1/2$ and the squares in the likelihood formula come from binary sign probabilities and squared reliabilities.

\begin{lemma}[Scalar necessary conditions]\label{lem:constraints}
For general integers $\ell>r\geq2$ and $g\geq2$, every fixed point with $C(c)>R$ satisfies
\begin{align}
m&\geq\Rk_{\ell-1}(m^{r-1}n^g),\label{eq:cm}\\
t&\geq\max\{R,f(n),\mathcal T_{g-1,d_0}(m^rn^{g-1})\},\quad t\leq n,\quad t>R,\label{eq:ct}\\
H(a)&\leq\min\{1-f(m),(1-m^{r-1}n^g)^{(\ell-1)/2}\}.\label{eq:ch}
\end{align}
\end{lemma}
\begin{proof}
At a fixed point of~\eqref{eq:DE}, apply Lemma~\ref{lem:BSC} to $a=q_1^{\star(\ell-1)}$. The check-moment product identity~\cite[Proposition 6(iii)]{KumarYoungMacrisPfister2014} applied to~\eqref{eq:q} gives $M_1(q_1)=m^{r-1}n^g$, proving~\eqref{eq:cm}.
For $b=c\star q_2^{\star(g-1)}$, use the information comparison in Lemma~\ref{lem:BSC} to replace each of the $g$ observations by a BSC with the same mutual information. Reduce the channel replacement information from $C(c)$ to $J(d_0)<R$, and each check replacement information from $I(q_2)$ to $f(m^rn^{g-1})$, using~\eqref{eq:fbound} and the same check-moment product identity. This gives~\eqref{eq:ct}. Equations~\eqref{eq:fbound} and~\eqref{eq:HZ} give
$H(a)\leq Z(q_1)^{\ell-1}\leq(1-m^{r-1}n^g)^{(\ell-1)/2}$ and the other entropy bound.
\end{proof}
These inequalities enclose the statistics of every genuine fixed point. Satisfying them alone does not establish the existence of a density-evolution fixed point with those statistics. This is the necessary-condition reduction used in place of the exact fixed-point parameterization available on the BEC.

\subsection{Bounds on the two-density functional}
The functional $S$ depends on all message moments. The maximization below bounds the terms subtracted in its moment expansion, giving a lower bound using only $m,n,t$. Combined with the positive term in~\eqref{eq:decomp}, these bounds will establish potential positivity by analytic inequalities or finite certificates.

For fixed $r,g$, define
\begin{align}
h_t(A,B)&=(r-t)A^{r-2}B^g+(g-1)A^{r-1}B^{g-1}\nonumber\\
&\quad -(r+g-2-t)A^{r-1}B^g,\label{eq:h}\\
G_t(m,n)&=\max_{0\leq A\leq m}h_t(A,n).
\end{align}
Substituting the entropy series and check-moment product identity~\cite[Propositions 6(iii) and 7]{KumarYoungMacrisPfister2014} into~\eqref{eq:S}, and collecting terms according to~\eqref{eq:h}, gives
\begin{equation}\label{eq:Smoment}
S=\sum_{k\geq1}\gamma_k A_k[t-h_t(A_k,B_k)].
\end{equation}
Zeroth scalar powers are one, including at zero. Direct differentiation of~\eqref{eq:h} gives, on the unit cube,
\begin{align}
\partial_Bh_t&=A^{r-2}B^{g-2}\bigl\{(g-1)^2A(1-B)\nonumber\\
&\quad +g(r-t)B(1-A)+AB\bigr\}\geq0,\label{eq:hmono}\\
\partial_t(t-h_t)&=1+A^{r-2}(1-A)B^g\geq1.\nonumber
\end{align}
Every term in braces is nonnegative. Consequently
\begin{equation}\label{eq:Gbound}
S\geq I(a)[t-G_t(m,n)].
\end{equation}

Indeed, $0\leq A_k\leq m$ and $0\leq B_k\leq n$ for every $k$, so $h_t(A_k,B_k)\leq G_t(m,n)$. Applying this bound to each summand in~\eqref{eq:Smoment} leaves $\sum_k\gamma_kA_k=I(a)$ by~\eqref{eq:moments}. The higher moments remain unrestricted except for these valid bounds; no choice of a BMS density is discarded by this maximization.

For exact evaluation of $G_t$, put
\[
u=(r-t)n^g,\qquad v=n^{g-1}[(g-1)-(r+g-2-t)n]
\]
Equation~\eqref{eq:h} becomes $h_t(A,n)=A^{r-2}(u+vA)$, and a maximizer on $[0,m]$ is
\[
A_*=
\begin{cases}
m,&v\geq0,\\
\min\{m,-(r-2)u/((r-1)v)\},&v<0.
\end{cases}
\]
Indeed, $u\geq0$. For $r>2$, the derivative has the sign of $(r-2)u+(r-1)vA$ on $A>0$; continuity handles zero. For $r=2$, the function is $u+vA$, so $v<0$ gives $A_*=0$, also as in the formula.

To display the degree dependence, write $h_{r,g,t}=h_t$. Direct subtraction in~\eqref{eq:h} gives
\begin{align*}
h_{r+1,g,t}-h_{r,g,t}
&=-A^{r-2}(1-A)\bigl[(r-t)(1-A)B^g\\
&\qquad+A(g-1)(1-B)B^{g-1}\bigr]\leq0,\\
h_{r,g+1,t}-h_{r,g,t}
&=-A^{r-2}B^{g-1}(1-B)\bigl[(r-t)(1-A)B\\
&\qquad+A(g-1)(1-B)\bigr]\leq0.
\end{align*}
Here the arguments are $(A,B)\in[0,1]^2$ and $t\in[0,1]$. Thus, for fixed BMS densities $a,b$, the moment representation~\eqref{eq:Smoment} gives the degree comparison
\begin{equation}\label{eq:degreecomparison}
S_{r+1,g}(a,b)\geq S_{r,g}(a,b),\qquad
S_{r,g+1}(a,b)\geq S_{r,g}(a,b).
\end{equation}

This comparison keeps the density pair fixed. It therefore extends inequalities for $S$ that hold for arbitrary density pairs. The set of fixed points also changes with the degrees, so each use of fixed-point constraints retains its stated degree assumptions.

Two estimates handle opposite parts of the remaining problem. When the transmitted-message moment $n$ is small, maximizing over the punctured-message moments gives a nonnegative bound for $S$. When the punctured-message entropy $H(a)$ is small, comparison with $a=\Delta_\infty$ controls the loss in $S$ by that entropy. The next two lemmas quantify these statements.

\begin{lemma}[Mean squared reliability of transmitted messages]\label{lem:weaktrans}
Let $r\geq2$, $g\geq2$, and let $a,b$ be arbitrary BMS densities. If
\[
n\leq\frac{(r-1)(g-1)}{(r-1)(g-1)+1}
\]
then
\begin{equation}\label{eq:weaktransgeneral}
S\geq I(a)\{f(n)-(g-1)n^{g-1}+(g-2)n^g\}.
\end{equation}
For $r\geq2$ and $g\geq3$, the following quantitative bound holds:
\begin{equation}\label{eq:weakb}
n\leq14/25\ \Longrightarrow\ S\geq\frac{144}{390625}I(a)n\geq0.
\end{equation}
\end{lemma}
\begin{proof}
Under the premise, $h_t(A,n)$ increases on $A\in[0,1]$. If $v\geq0$, this follows from $u\geq0$. If $v<0$ and $r>2$, the bracket determining the sign of the derivative is minimized at $A=1$, where
\[
(r-2)u+(r-1)v\geq n^{g-1}\{(r-1)(g-1)-[(r-1)(g-1)+1]n\}\geq0.
\]
For $r=2$, the premise itself gives $v\geq0$. Since $h_t(1,n)=n^{g-1}[(g-1)-(g-2)n]$, equations~\eqref{eq:Gbound} and~\eqref{eq:fbound} give $S\geq I(a)[t-h_t(1,n)]$ and $t\geq f(n)$, proving~\eqref{eq:weaktransgeneral}.

For the quantitative bound, the degree comparison~\eqref{eq:degreecomparison} reduces $g\geq3$ to $g=3$, with $r\geq2$ arbitrary. The cutoff in the premise is then at least $2/3>14/25$.
Using $\ln2<25/36$ from Appendix~\ref{app:elementary} and retaining the first four positive terms of the series~\eqref{eq:moments}~\cite[Proposition 7]{KumarYoungMacrisPfister2014} for $f$ gives
\[
\begin{aligned}
\frac{f(n)-n^{g-1}[(g-1)-(g-2)n]}{n}&\geq P_*(n),\\
P_*(n)&:=\frac{36}{25}\sum_{j=1}^{4}\frac{n^{j-1}}{2j(2j-1)}
-(g-1)n^{g-2}+(g-2)n^{g-1}.
\end{aligned}
\]
The derivative $P_*'$ increases and
$P_*'(14/25)=-10657/15625<0$. Consequently
$P_*(n)\geq P_*(14/25)=144/390625>0$.
Continuity includes $n=0$.
\end{proof}

Define
\[
\rho_g(B)=B-(g-1)B^{g-1}+(g-2)B^g.
\]
The finite geometric-sum identity gives
\begin{equation}\label{eq:rhog}
\rho_g(B)=B(1-B)^2\sum_{j=1}^{g-2}jB^{j-1}\geq0
\qquad(0\leq B\leq1).
\end{equation}
The sum is empty and equals zero when $g=2$.
\begin{lemma}[Comparison with perfect messages from punctured variables]\label{lem:perfect}
For $r\geq2$, $g\geq2$, and arbitrary BMS densities $a,b$,
\begin{equation}\label{eq:nearperfect}
S(a,b)\geq S(\Delta_\infty,b)-I(b)H(a)
\geq\frac7{10}\rho_g(n)-H(a).
\end{equation}
\end{lemma}
\begin{proof}
Compare each term of~\eqref{eq:Smoment} with its value at $A_k=1$, which corresponds to perfect information in $a$. The linear term $tA_k$ is smaller by $t(1-A_k)$; by~\eqref{eq:moments}, their weighted sum is $I(b)H(a)$. The derivative below shows that the remaining part is at least its value at $A_k=1$, proving the first bound.

Put $s_t(A,B)=A[t-h_t(A,B)]$ and $N_t=s_t-tA$. Differentiating~\eqref{eq:h} gives
\[
\begin{aligned}
\partial_AN_t={}&-A^{r-2}B^{g-1}\bigl\{r(g-1)A(1-B)\\
&\qquad+(r-1)(r-t)B(1-A)+tAB\bigr\}\leq0.
\end{aligned}
\]
Every term in braces is nonnegative, so
\[
N_t(A,B)\geq N_t(1,B)=-(g-1)B^{g-1}+(g-2)B^g.
\]
Adding $tA$ and summing according to~\eqref{eq:Smoment} proves the first bound. Equations~\eqref{eq:Smoment} and~\eqref{eq:rhog}, together with the logarithm estimate in Appendix~\ref{app:elementary}, give
\[
S(\Delta_\infty,b)=\sum_k\gamma_k\rho_g(B_k)
\geq\gamma_1\rho_g(n),\qquad\gamma_1>7/10.
\]
Now use $I(b)\leq1$.
\end{proof}

\subsection{Mean squared reliability of messages from punctured variables}\label{sec:specialized}
The next bound handles the entire region $m\leq4/5$ for $r,g\geq3$. It gives $S\geq0$ for arbitrary density pairs in this region, and hence $U>0$ at its nontrivial fixed points when $C(c)>R$, by~\eqref{eq:decomp}. Only the region with larger $m$ remains to be checked.

\begin{lemma}[Mean squared reliability of messages from punctured variables]\label{lem:weak}
For $r,g\geq3$ and arbitrary BMS pairs,
\begin{equation}
m\leq4/5\ \Longrightarrow\ S\geq I(a\cc b)/1000\geq I(a)I(b)/1000\geq0,\label{eq:weaka}
\end{equation}
\end{lemma}
\begin{proof}
The degree comparison~\eqref{eq:degreecomparison} reduces the proof to $r=g=3$. Put $a_0=4/5$. We retain $r,g$ in degree-derived coefficients; the rational cutoffs and margins are chosen for this local specialization.

We seek the termwise bound $t-h_t(A_k,B_k)\geq B_k/1000$. Since $A_k\leq a_0$, monotonicity first removes $A_k$ by replacing it with $a_0$. The lower bound $t\geq f(B_k)$ from~\eqref{eq:fbound} and $B_k\leq n$ then leaves a function of $B_k$ alone. A lower bound on its curvature will prove positivity over the entire interval.

When $A\leq a_0$ and $t\geq f(B)\geq(7/10)B$, $h_t(A,B)$ increases with $A$. Differentiating~\eqref{eq:h} gives
\[
\partial_Ah_t=A^{r-3}B^{g-1}\bigl[(r-2)(r-t)B
+(r-1)A((g-1)-(r+g-2-t)B)\bigr].
\]
If the coefficient of $A$ in the bracket is negative, its minimum occurs at $A=a_0$. There the coefficient of $t$ is $[(r-1)a_0-(r-2)]B\geq0$. Replacing $t$ by $(7/10)B$ gives a decreasing polynomial in $B$, whose value at $B=1$ is
\[
(r-2)(r-7/10)+(r-1)a_0[(g-1)-(r+g-2-7/10)]=11/50>0.
\]
Define the coefficients inherited from $h_t(a_0,B)$ by
\begin{align*}
\alpha_0&=a_0^{r-2}(1-a_0),&\beta_0&=(g-1)a_0^{r-1},\\
\chi_0&=(r+g-2)a_0^{r-1}-r a_0^{r-2}.
\end{align*}
Therefore
\[
t-h_t(A,B)\geq(1+\alpha_0B^g)f(B)-\beta_0B^{g-1}+\chi_0B^g.
\]
For $0<B<1$, define
\begin{equation}\label{eq:Kdefinition}
K(B)=\frac{(1+\alpha_0B^g)f(B)-\beta_0B^{g-1}+\chi_0B^g}{B}.
\end{equation}
Substitution of the series~\eqref{eq:moments}~\cite[Proposition 7]{KumarYoungMacrisPfister2014} into~\eqref{eq:Kdefinition} gives
\[
K(B)=\sum_{k\geq1}\gamma_kB^{k-1}
+\alpha_0\sum_{k\geq1}\gamma_kB^{k+g-1}
-\beta_0B^{g-2}+\chi_0B^{g-1}.
\]
The negative term is linear under the stated specialization, and all other coefficients are nonnegative. Using $\ln2<5/6$ from Appendix~\ref{app:elementary}, we have $\gamma_g=1/[2g(2g-1)\ln2]>1/25$,
and hence $K''(B)\geq(g-1)(g-2)(\gamma_g+\chi_0)B^{g-3}>2/5$.
At $q=441/484$, the elementary bounds in Appendix~\ref{app:elementary} give
\[
K(q)\geq\frac7{2000},\qquad |K'(q)|\leq\frac1{50}.
\]

The rational point $q$ is useful because $\sqrt q=21/22$ permits elementary logarithm bounds, and the derivative there is small. A positive value at this single point is supplemented by the uniform curvature bound: together they give a quadratic lower bound valid for every $B$. Its minimum is controlled by completing the square.
Integrating the second-derivative bound gives a quadratic strengthening of the tangent lower bound for a convex function~\cite[Section 3.1.3]{BoydVandenberghe2004}. Completing the square yields
\[
K(B)\geq K(q)+K'(q)(B-q)+\frac{(B-q)^2}{5}
\geq\frac7{2000}-\frac54\left(\frac1{50}\right)^2
=\frac3{1000}>\frac1{1000}.
\]
Continuity includes both endpoints. Thus
\[
(1+\alpha_0B^g)f(B)-\beta_0B^{g-1}+\chi_0B^g\geq B/1000
\qquad(0\leq B\leq1).
\]
Apply this inequality to each $(A_k,B_k)$ in~\eqref{eq:Smoment}
and use the covariance argument following~\eqref{eq:productH}, which gives $I(a\cc b)\geq I(a)I(b)$.

\end{proof}

\subsection{All sufficiently large punctured-variable degrees}\label{sec:large-degrees}
Set $r=g=3$ and $a_0=4/5$, so $\rho_g(n)=n(1-n)^2$. The preceding bounds establish positivity when either message moment is small. For $\ell\geq33$, convexity of the necessary entropy bound handles the remaining region: it either excludes a nontrivial fixed point or makes its potential positive. This proves positivity at every nontrivial fixed point in this degree range when $C(c)>R$.

\begin{proposition}\label{prop:large}
For every integer $\ell\geq33$ and every $c$ with $C(c)>r/\ell$,
each fixed point other than $\gd$ has $U>0$.
\end{proposition}
\begin{proof}
Lemmas~\ref{lem:weak} and~\ref{lem:weaktrans}, together with~\eqref{eq:decomp}, handle $m\leq4/5$ or $n\leq14/25$.
Put $p=(\ell-1)/2\geq16$ and, for the remaining region, define
\[
F_n(m)=(1-m^{r-1}n^g)^p,\qquad
h(n)=\frac7{10}\rho_g(n)=\frac7{10}n(1-n)^2,\qquad m_*(n)=1-h(n).
\]
Equations~\eqref{eq:fbound} and~\eqref{eq:ch} require $1-m\leq H(a)\leq F_n(m)$ at a fixed point.

For fixed $n<1$, the choice $m_*=1-h(n)$ divides the task at a point where $1-m_*=h(n)$. On $[a_0,m_*]$, convexity and bounds at the two endpoints will give $F_n(m)<1-m$, contradicting the necessary entropy bounds. On $[m_*,1]$, decreasing $F_n$ below $h(n)$ will instead give $S>0$ by Lemma~\ref{lem:perfect}. Thus the calculation below only needs convexity and the two endpoint bounds $F_n(a_0)<1-a_0$ and $F_n(m_*)<h(n)$. Increasing $\ell$ decreases $F_n$, which is why one lower cutoff for $p$ handles all larger degrees.
Differentiating the displayed definition of $F_n$ twice gives, for $m\geq4/5$ and $n\geq14/25$,
\[
\begin{aligned}
F_n''(m)={}&p(r-1)m^{r-3}n^g(1-m^{r-1}n^g)^{p-2}\\
&\quad\cdot\bigl\{[(r-1)p-1]m^{r-1}n^g-(r-2)\bigr\}\geq0,
\end{aligned}
\]
since $[(r-1)p-1]m^{r-1}n^g\geq[(r-1)16-1]a_0^{r-1}(14/25)^g>r-2$.
Moreover,
\begin{equation}\label{eq:tail1}
F_n(a_0)\leq[1-a_0^{r-1}(14/25)^g]^{16}=(346721/390625)^{16}
<(9/10)^{16}<(2/3)^4=16/81<1/5.
\end{equation}
Here $(9/10)^4=6561/10000<2/3$.

We next prove $F_n(m_*)<h(n)$ for $14/25\leq n<1$.
Since $h$ decreases on this interval and
$h(14/25)=5929/78125<1/13$, we have $m_*>12/13$.
Put $c_0=(12/13)^{r-1}$. For $14/25\leq n\leq4/5$, the ratio
\[
\frac{(1-c_0n^g)^p}{n(1-n)^2}
\]
is decreasing. Indeed, direct differentiation gives
\[
\frac{d}{dn}\log\frac{(1-c_0n^g)^p}{n(1-n)^2}
=-\frac{p g c_0n^{g-1}}{1-c_0n^g}
+\frac{3n-1}{n(1-n)}.
\]
All denominators are positive on this interval. Since $3n-1\leq7/5$, $1-c_0n^g\leq1$, $n\geq14/25$, and $1-n\geq1/5$, we have
\[
\frac{(3n-1)(1-c_0n^g)}{g c_0n^g(1-n)}
\leq\frac{7/5}{g c_0(14/25)^g(1/5)}
=\frac{2640625}{169344}<16\leq p.
\]
This inequality makes the logarithmic derivative strictly negative. Here $3n-1$ comes from differentiating the polynomial $n(1-n)^2$; its coefficient is independent of the node degree $g$. At the left endpoint the base is $x_0=1-c_0(14/25)^g=2245489/2640625<1701/2000$.
The successive squaring bounds in Appendix~\ref{app:elementary} give
\begin{equation}\label{eq:tail2}
x_0^{16}<\frac{5929}{78125}=h(14/25).
\end{equation}
These comparisons prove $F_n(m_*)<h(n)$ throughout this interval.

For $a_0\leq n<1$ we have $m_*>n$. Put $b_0=1-a_0^{r+g-1}<7/10$.
The finite geometric sum gives $1-n^{r+g-1}\leq(r+g-1)(1-n)$. Splitting off two powers and using $n\geq a_0$ and $1-n^{r+g-1}\leq b_0$ gives
\[
\frac{F_n(m_*)}{n(1-n)^2}
\leq\frac{(1-n^{r+g-1})^p}{n(1-n)^2}
\leq\frac{(r+g-1)^2}{a_0}b_0^{p-2}
<\frac{(r+g-1)^2}{a_0}(7/10)^{12}
<\frac{125}{256}<\frac7{10}.
\]
Here $p-2\geq14$ and $(7/10)^4<1/4$. This completes the second endpoint bound.

The chord upper bound for a convex function~\cite[Section 3.1.1]{BoydVandenberghe2004} and the strict endpoint inequalities now imply
$F_n(m)<1-m$ for every $4/5\leq m\leq m_*$, excluding a fixed point
there. For $m\geq m_*$, monotonicity instead gives
$H(a)\leq F_n(m)\leq F_n(m_*)<h(n)$, so Lemma~\ref{lem:perfect} yields $S>0$.
Finally, if $n=1$, then $F_1(1)=0$; convexity and the strict bound at
$4/5$ give $F_1(m)<1-m$ for $4/5\leq m<1$.
The remaining corner $(m,n)=(1,1)$ is the perfect-information fixed point.
Equation~\eqref{eq:decomp} therefore gives $U>0$ at every fixed point other than $\gd$.
\end{proof}

\section{Finite certificates and positive fixed-point gaps}\label{sec:certificates}
For $r=g=3$ and $C(c)>R$, the analytic bounds have settled all $\ell\geq33$. The remaining task is positivity at $x\in\Fnt(c)$, defined in~\eqref{eq:fixedpoint-sets}, for $4\leq\ell\leq32$, which we complete using exact interval certificates. Together with the known value $U(\Delta_0,c;c)=C(c)-R>0$, this proves positivity on the whole set $\Fneq(c)$. Separation from $\gd$ and compactness then turn these signs into the positive gap needed for the coupling argument in Section~\ref{sec:saturation}.

\paragraph*{Why interval verification gives a rigorous proof}
Interval arithmetic gives lower and upper bounds that hold for every real point of an interval or rectangle~\cite[Chapters 2--5]{MooreKearfottCloud2009}. A strictly positive lower bound therefore proves positivity throughout the region. For example, if $x\in[1/2,3/4]$, then $1-x\in[1/4,1/2]$, so
\[
\frac18\leq x(1-x)\leq\frac38
\qquad\text{for every }x\in[1/2,3/4].
\]
These rational inequalities prove positivity on the entire interval, including its endpoints. The bounds need not be sharp.

In the computations below, all interval endpoints are rational. When rounding is needed, the lower endpoint is rounded downward and the upper endpoint upward, so the interval continues to contain every possible exact value. Logarithm evaluations include a proved bound for the omitted series terms. Appendix~\ref{app:arithmetic} justifies the enclosure rules for every operation; thus rounding and truncation errors are included in each sign test. An enclosure containing zero does not by itself decide strict positivity. Subdivision can tighten an inconclusive bound.

Here the analytic reductions enclose all fixed points other than $\gd$ by scalar necessary conditions. Each remaining closed rectangle is then shown either to contain no such fixed point or to imply positive potential at every such fixed point it contains. The certificate records the rectangles and their justifications. Verifying both the inequalities and a complete cover, including all shared boundaries, proves the finite sign conditions rigorously.

A positive value at each fixed point in $\Fneq(c)$ does not by itself give a positive infimum: positive values could approach zero along a sequence. We first prove that $\gd$, where $U=0$, has a sufficiently small neighborhood containing no other fixed point. The separation lemma establishes compactness of $\Fneq(c)$ under its stated general-degree assumptions, covering both the second trivial fixed point and all points of $\Fnt(c)$. Once positivity is established, continuity gives an attained, strictly positive minimum.

\begin{lemma}[Separation from perfect information]\label{lem:separation}
Fix integers $\ell>r\geq2$ and $g\geq2$. If $g\geq3$, put $d=r+g-1$ and $\mu=\min\{\ell-1,g-1\}>1$. For every BMS density $c$, every uncoupled fixed point distinct from $\gd$ then satisfies
\begin{equation}\label{eq:separation}
\max\{Z(a),Z(b)\}\geq d^{-\mu/(\mu-1)}.
\end{equation}
If $g=2$ and $C(c)>0$, put $z=Z(c)<1$ and
\[
B_{r,g}(z)=(r-1)+\frac{grz}{1-z}.
\]
Every uncoupled fixed point distinct from $\gd$ satisfies
\begin{equation}\label{eq:separationg2}
Z(a)\geq B_{r,g}(z)^{-(\ell-1)/(\ell-2)}.
\end{equation}
In either case, $\Fneq(c)$ is nonempty and compact, and $\gap(c)=\min_{x\in\Fneq(c)}U(x;c)$. Let $\mathcal C$ be a nonempty compact family of BMS channels; when $g=2$, also assume $C(c)\geq C_0>0$ for all $c\in\mathcal C$. Then the set $\{(c,x):c\in\mathcal C,\ x\in\Fneq(c)\}$ is compact. If $U(x;c)>0$ throughout this set, then $\inf_{c\in\mathcal C}\gap(c)>0$.
\end{lemma}
\begin{proof}
Near perfect information, check updates increase a small Bhattacharyya parameter by at most a constant factor, while variable updates raise it to a power greater than one. A sufficiently small positive value must therefore decrease strictly and cannot be a fixed point. For $g=2$, the transmitted update has only one incoming check message, so we instead use the strict factor $Z(c)<1$ supplied by the channel.

For $g\geq3$, put $\eta=\max\{Z(a),Z(b)\}$. Each check update has $d$ inputs, and each variable update has at least $\mu$ incoming check messages. If $\eta<d^{-\mu/(\mu-1)}$, then $d\eta<1$, and applying the Bhattacharyya bounds~\eqref{eq:HZ} to the $d$ check inputs in~\eqref{eq:q} and the variable updates in~\eqref{eq:DE} gives $\eta(\DE(a,b))\leq(d\eta)^\mu$. The multiplicative variable-update identity is given in~\cite[Section II-D]{KumarYoungMacrisPfister2014}. A fixed point with $0<\eta<d^{-\mu/(\mu-1)}$ would satisfy $1\leq d^\mu\eta^{\mu-1}<1$. The case $\eta=0$ is exactly $\gd$.

For $g=2$, put $\alpha=Z(a)$ and $\beta=Z(b)$. Equations~\eqref{eq:HZ} and~\eqref{eq:fbound} give $z\leq\sqrt{1-C(c)}<1$. Applying the same bounds~\eqref{eq:HZ} to~\eqref{eq:q} and~\eqref{eq:DE}, now retaining the factor $z$ in the transmitted update, gives at a fixed point
\[
\alpha\leq[(r-1)\alpha+g\beta]^{\ell-1},\qquad
\beta\leq z[r\alpha+(g-1)\beta].
\]
Since $g-1=1$, these imply $\beta\leq rz\alpha/(1-z)$ and $\alpha\leq[B_{r,g}(z)\alpha]^{\ell-1}$. If $\alpha=0$, then $\beta=0$ and the fixed point has perfect information. Otherwise division by $\alpha>0$ proves~\eqref{eq:separationg2}.

Represent symmetric densities by probability laws on $[0,1]$ of their absolute reliabilities. This space is compact in the weak topology~\cite[Proposition 65 and Corollary 66]{KumarYoungMacrisPfister2014}. The entropy distance~\cite[Definition 10]{KumarYoungMacrisPfister2014}
$\sum_k\gamma_k|M_k(x)-M_k(y)|$ gives the same topology: weak convergence gives convergence of moments and dominated convergence of the series; conversely all moments of $D^2$ determine the law on this compact interval. Entropy and $Z$ are expectations of continuous bounded kernels. Check convolution is continuous by multiplication of reliabilities. For variable convolution, adding the two input LLRs and conditioning on agreement or disagreement of their signs gives the two output reliabilities; this is the kernel in~\cite[proof of Proposition 68]{KumarYoungMacrisPfister2014}, written in reliability coordinates:
\[
\frac{u+v}{1+uv},\qquad\frac{|u-v|}{1-uv},
\]
with weights $(1+uv)/2$ and $(1-uv)/2$. The second weight tends to zero at the sole singular corner $(1,1)$, so averaging any continuous bounded test function is continuous there as well. Hence both convolutions and $U$ are continuous. These are the compactness and continuity properties established in~\cite[Proposition 11(ii)--(iv) and Appendix I]{KumarYoungMacrisPfister2014}.

The fixed-point set is closed. The subset $\Fneq(c)$ obtained by removing $\gd$ is exactly the closed subset satisfying the applicable separation bound, \eqref{eq:separation} or~\eqref{eq:separationg2}, and is therefore compact. It is nonempty because $(\Delta_0,c)$ is a fixed point distinct from $\gd$, including when $c$ is noiseless. The continuous potential therefore attains its minimum there. For $g=2$ and $C(c)\geq C_0>0$, we have $z\leq z_0:=\sqrt{1-C_0}<1$. Since $B_{r,g}$ increases with $z$, equation~\eqref{eq:separationg2} gives the uniform lower bound $Z(a)\geq B_{r,g}(z_0)^{-(\ell-1)/(\ell-2)}>0$. For a nonempty compact channel family $\mathcal C$ satisfying the stated assumptions, the same closed-set argument applies jointly to $(c,x)$, proving compactness of $\{(c,x):c\in\mathcal C,\ x\in\Fneq(c)\}$. If $U(x;c)>0$ throughout this nonempty set, its attained minimum is strictly positive.
\end{proof}

For the finite-certificate and positive-gap results in the remainder of this section, set $r=g=3$. Then $\mu=g-1=2$. For the remaining degrees $4\leq\ell\leq32$, use the initial rectangle at the root of the subdivision tree
\begin{equation}\label{eq:root}
\mathcal B_\ell=[4/5,1]\times[n_0,1],\qquad n_0=\max\{11/20,R\}.
\end{equation}
Here $m=M_1(a)$ and $n=M_1(b)$ are the scalar coordinates introduced in Section~\ref{sec:bounds}; $f$ is defined in~\eqref{eq:scalar-information}.

This rectangle overlaps the analytically treated regions. Indeed, Lemma~\ref{lem:weaktrans} reaches $14/25>11/20$, and~\eqref{eq:ct} gives $n\geq t>R$ at every fixed point under consideration. Together with Lemma~\ref{lem:weak} for $m\leq4/5$, these observations show that no fixed point is lost when verification is restricted to $\mathcal B_\ell$.
The 29 degree records used below contain 3,599 leaf rectangles. Degree-dependent expressions in the rules are evaluated at $r=g=3$. The published archive also retains certificates for $33\leq\ell\leq46$; Proposition~\ref{prop:large} treats those degrees analytically.
For a closed rectangle $B=[m_-,m_+]\times[n_-,n_+]$, put $k=\ell-1$ and
\begin{align}
q_-&=m_-^{r-1}n_-^g,&q_+&=m_+^{r-1}n_+^g,&q_{2,-}&=m_-^rn_-^{g-1},\nonumber\\
\tau&=\max\{R,\underline f(n_-),\underline{\mathcal T_{g-1,d_0}(q_{2,-})}\},&&&\label{eq:boxdefs}\\
\overline H_a&=\min\{1-\underline f(m_-),\overline{(1-q_-)^{k/2}}\}.&&&\nonumber
\end{align}
Underlines and overlines mean certified lower and upper bounds. The interval for $J(d_0)$ has upper endpoint strictly less than $R$. The monotone endpoint bounds in~\eqref{eq:ct} and~\eqref{eq:ch} therefore imply $t\geq\tau$ and $H(a)\leq\overline H_a$ at any fixed point in $B$ for a channel whose capacity exceeds the rate.

The finite-degree task is to certify the sign conditions left by the scalar reduction. In the BEC proof, this task is completed by a Sturm root count for a polynomial~\cite{OkazakiKasai2014}; here it is completed by the following exclusion and positivity tests on closed rectangles. The checker tests them in order, with every comparison valid over the entire rectangle.
\begin{enumerate}
\item \emph{Exclusion near perfect information:} $m_-,n_->1-d^{-2\mu/(\mu-1)}$. Then $Z(a),Z(b)<d^{-\mu/(\mu-1)}$ by~\eqref{eq:HZ}, and Lemma~\ref{lem:separation} excludes every fixed point other than $\gd$.
\item \emph{Entropy bound for positivity of $S$:} $(7/10)\bigl[n_+-(g-1)n_+^{g-1}+(g-2)n_+^g\bigr]>\overline H_a$. Under $g=3$, the function $n-(g-1)n^{g-1}+(g-2)n^g=n(1-n)^2$ decreases for $n\geq11/20$, so~\eqref{eq:nearperfect} gives $S>0$.
\item \emph{Exclusion by the punctured-variable update:} $\underline{\Rk_k(q_-)}>m_+$. Monotonicity contradicts~\eqref{eq:cm}.
\item \emph{Transmitted-information exclusion:} $\tau>n_+$. This contradicts $t\leq n$ in~\eqref{eq:ct}.
\item \emph{Nonnegative $S$:} $z=\tau-G_\tau(m_+,n_+)\geq0$. The previous test guarantees $0\leq\tau\leq n_+\leq1$, so~\eqref{eq:hmono} and the quadratic maximum~\eqref{eq:Gbound} apply. They give $S\geq I(a)z\geq0$.
\item \emph{Positive potential:} $z<0$ and
\begin{equation}\label{eq:positiveU}
m_+z+(\tau-R)(1-m_+)(1-q_+)>0.
\end{equation}
Indeed, $S\geq I(a)z\geq m_+z$, where $I(a)\leq m\leq m_+$ follows from~\eqref{eq:fbound} and the second inequality reverses because $z<0$. For the other term, apply~\eqref{eq:productH}, then the lower entropy bound in~\eqref{eq:HZ} and the check-moment identity~\eqref{eq:moments}:
\[
D\geq H(a)H(q_1)\geq(1-m_+)(1-q_+),\quad t-R\geq\tau-R\geq0.
\]
Equation~\eqref{eq:decomp} then proves $U>0$.
\end{enumerate}
For rule 5 and rule 6, one can apply~\eqref{eq:hmono} separately to each summand in~\eqref{eq:Smoment}: lower $t$ to $\tau$, raise $B_k$ to $n_+$, and maximize over $A_k\in[0,m_+]$. Every rule proving $S\geq0$ yields strict $U>0$ at a fixed point other than $\gd$ through $(t-R)D>0$, as established in Lemma~\ref{lem:reduction}. In particular, the sign of the lower estimate on $S$ in rule 6 is not an assumption about the actual sign of $S$.

\begin{proposition}[Finite certificate]\label{prop:finite}
The certificate records for $\ell=4,\ldots,32$ classify a complete closed-rectangle cover of~\eqref{eq:root} under the six rules above. They have 3,599 leaves, maximum binary depth 20, and zero unclassified leaves. Consequently, for these degrees and a channel whose capacity exceeds the rate, every fixed point other than $\gd$ has positive potential.
\end{proposition}
\begin{proof}
The standard-library program and complete rational leaf data in the versioned supplement~\cite{MNBMSCertificates2026} form the computational part of this proof. The arithmetic and coverage specifications appear in Appendix~\ref{app:arithmetic}. Verification reconstructs each rectangle from its binary path, checks the specified rule and its exact stored margin, and verifies a full prefix-tree cover. Table~\ref{tab:rules} summarizes the 29 records used here; per-degree counts and reproduction instructions are given in the supplement. The mathematical validity of each classification is established above. Lemmas~\ref{lem:weak} and~\ref{lem:weaktrans}, together with the necessary condition~\eqref{eq:ct}, cover the remaining scalar region.
\end{proof}
\begin{table}[htbp]
\centering
\caption{Verified leaf counts by rule for the finite range $4\leq\ell\leq32$ used in the proof.}\label{tab:rules}
\begin{tabular}{lr}\toprule
Rule & Leaves\\\midrule
Exclusion near perfect information &30\\
Entropy bound for positivity of $S$ &498\\
Exclusion by the punctured-variable update &1073\\
Transmitted-information exclusion &60\\
Nonnegative $S$ &576\\
Positive potential &1362\\\midrule
Total &3599\\\bottomrule
\end{tabular}
\end{table}

\begin{theorem}[Positive fixed-point gap]\label{thm:gap}
For $r=g=3$, all integers $\ell\geq4$, and all BMS densities $c$ with $C(c)>R$,
\begin{equation}\label{eq:gap}
\gap(c)=\min_{x\in\Fneq(c)}U(x;c)>0.
\end{equation}
\end{theorem}
\begin{proof}
Propositions~\ref{prop:large} and~\ref{prop:finite} prove pointwise positivity at every fixed point other than $\gd$. Lemma~\ref{lem:separation} shows that $\Fneq(c)$ is nonempty and compact, so the continuous potential attains its minimum there. That minimum is strictly positive.
\end{proof}

\begin{corollary}[Uniform positive fixed-point gap]\label{cor:uniform}
For $r=g=3$, fixed $\ell\geq4$, and $0<\delta\leq1-R$, the infimum of $\gap(c)$ over $C(c)\geq R+\delta$ is positive.
\end{corollary}
\begin{proof}
The channel family $\mathcal C=\{c:C(c)\geq R+\delta\}$ is nonempty and is a closed subset of the compact density space because $C(c)=1-H(c)$ is continuous~\cite[Proposition 11(ii)--(iii)]{KumarYoungMacrisPfister2014}. Theorem~\ref{thm:gap} gives positivity at every fixed point other than $\gd$ throughout this family. The uniform conclusion of Lemma~\ref{lem:separation} therefore applies.

Compactness is used jointly for the channel and the fixed point. If the gaps could approach zero as channels varied, a convergent subsequence of such pairs would give a fixed point distinct from $\gd$ with zero potential on a channel in the same family, contradicting Theorem~\ref{thm:gap}.
\end{proof}

\section{Threshold saturation for the two-type recursion}\label{sec:saturation}
A positive fixed-point gap is a sufficient condition that turns the uncoupled potential comparison into successful coupled decoding. This section proves that implication for general $\ell>r\geq2$ and $g\geq2$, following the threshold-saturation step of the BEC argument~\cite{OkazakiKasai2014}. The infimum $\gap(c)$ in~\eqref{eq:gapdefinition} includes both $\Fnt(c)$ from~\eqref{eq:fixedpoint-sets} and the trivial point $(\Delta_0,c)$; only $\gd$ is excluded. The theorem below assumes $\gap(c)>0$, which Section~\ref{sec:certificates} has established for degree three when $C(c)>R$.
The comparison-and-shift argument follows the potential method for vector recursions~\cite[Section III]{YedlaJianNguyenPfister2012} and BMS densities~\cite[Section IV-A]{KumarYoungMacrisPfister2014}. We derive its two-type form below, holding the exterior density fixed during variation and deriving an explicit bound for the MN update.
Define
\begin{align}
Q(x)&=-(r+g-1)H(s(x))+rH(q_1(x))+gH(q_2(x)),\nonumber\\
F_c(y)&=R H(y_1^{\star\ell})+H(c\star y_2^{\star g}).\label{eq:QF}
\end{align}
The map $\mathsf g$ returns $(q_1,q_2)$, so
\[
F_c(\mathsf g(x))=R H(q_1(x)^{\star\ell})+H(c\star q_2(x)^{\star g}),
\qquad U(x;c)=Q(x)-F_c(\mathsf g(x)).
\]
The first three terms of~\eqref{eq:U} constitute $Q$, and its two subtracted posterior entropies constitute $F_c\circ\mathsf g$. At $\gd$, the endpoint rules~\eqref{eq:endpoint-convolutions} give $Q(\gd)=-(r+g-1)\cdot0+r\cdot0+g\cdot0=0$ and $F_c(\mathsf g(\gd))=R\cdot0+0=0$. For an arbitrary density pair $y$, both entropies in $F_c(y)$ are nonnegative and $R>0$, so $F_c(y)\geq0$.

The ordered-update comparison below follows the ordered-mixture potential-descent argument of~\cite[Lemma 46]{KumarYoungMacrisPfister2014}; we give its two-type density-valued form explicitly.
The right endpoint of a coupled fixed profile need not itself be an uncoupled fixed point. The comparison construction will give the weaker relation $p\preceq\DE(p)$. Iterating this relation moves toward a less informative uncoupled fixed point while decreasing the potential, so its fixed-point gap still bounds $U(p;c)$ from below.

\begin{lemma}[Descent for ordered updates]\label{lem:descent}
If $x\preceq\DE(x)$ or $\DE(x)\preceq x$, then $U(\DE(x);c)\leq U(x;c)$. In particular, if $p\neq\gd$ and $p\preceq\DE(p)$, then $U(p;c)\geq\gap(c)$.
\end{lemma}
\begin{proof}
Suppose first $x\preceq\DE(x)$ and put $z_s=(1-s)x+s\DE(x)$. Monotonicity of the convolutions~\cite[Proposition 3(i)]{KumarYoungMacrisPfister2014}, applied to~\eqref{eq:DE}, gives $z_s\preceq\DE(x)\preceq\DE(z_s)$. Each term in $d\mathsf g(z_s)/ds$ is a positive multiple of a degraded-minus-informative difference, while each residual $z_{s,\alpha}-\mathsf T_{c,\alpha}(z_s)$ has the opposite orientation. The signed variable-entropy inequality~\cite[Proposition 8(iii)]{KumarYoungMacrisPfister2014}, applied in~\eqref{eq:dU} with these opposite orientations, gives $dU(z_s;c)/ds\leq0$. For $\DE(x)\preceq x$, one instead has $\DE(z_s)\preceq\DE(x)\preceq z_s$ and both orientations reverse. Integration proves descent.

Starting from $p\preceq\DE(p)$ yields a sequence increasing in degradation order. Its moments are monotone and determine a unique limiting density pair by compactness; the degradation order is closed (equivalently the binary-experiment convex-order inequalities pass to weak limits)~\cite[Proposition 11(v)--(vi)]{KumarYoungMacrisPfister2014}. Continuity of the update~\cite[Proposition 11(iv)]{KumarYoungMacrisPfister2014} makes the limit a fixed point $x_\infty\succeq p$, which cannot equal $\gd$. Descent and continuity imply $U(p;c)\geq U(x_\infty;c)\geq\gap(c)$.
\end{proof}

The following bound for a binary tree of variable and check convolutions follows by iterating entropy duality and degradation monotonicity; related two-convolution bounds appear in~\cite[Proposition 9]{KumarYoungMacrisPfister2014}.
\begin{lemma}[Entropy bounds for repeated convolutions]\label{lem:tree}
Consider a finite binary tree whose operations are $\star$ or $\cc$, with probability densities at its leaves. Replacing a leaf $a$ by $b\succeq a$ increases output entropy by a number in $[0,H(b)-H(a)]$. For any probability densities $u,v$, if two distinct leaves are replaced by signed differences $a-b$ and $u-v$, then the absolute value of the resulting multilinearly extended entropy is at most $H(b)-H(a)$.
\end{lemma}
\begin{proof}
For convolution with a fixed probability density, both entropy increases are nonnegative by degradation~\cite[Proposition 3(i) and Section II-C]{KumarYoungMacrisPfister2014}, and their sum is the input increase by the signed duality identity~\cite[Proposition 5]{KumarYoungMacrisPfister2014}. Each is at most that increase. Apply this along the path to the root. With the other selected leaf fixed first to $u$ and then to $v$, the signed entropy changes from $a-b$ both belong to $[-(H(b)-H(a)),0]$. Subtracting them proves the second assertion by multilinearity. The other difference requires no order assumption.
\end{proof}

For the sufficient width bound, define
\begin{equation}\label{eq:Kgeneral}
\mathcal K_{\ell,r,g}=(r+g-1)^2\bigl[r(\ell-1)+g(g-1)\bigr].
\end{equation}

\begin{theorem}[Conditional two-type saturation]\label{thm:saturation}
For the general $(\ell,r,g)$ construction of Section~\ref{sec:model}, suppose $\gap(c)>0$. Then, for every finite $L$, the terminated recursion converges to $\gd$ whenever
\begin{equation}\label{eq:widthgeneral}
w>\frac{\mathcal K_{\ell,r,g}}{2\gap(c)}.
\end{equation}
\end{theorem}
\begin{proof}
We prove the theorem by assigning two bounds to the same one-position shift. If a non-perfect limiting profile existed, the exact boundary identity would force a potential decrease of at least $\gap(c)$. Smoothing makes neighboring densities close, and stationarity makes the first variation vanish; the resulting decrease is at most $\mathcal K_{\ell,r,g}/(2w)$. The width condition makes these requirements incompatible.

Put $N=L+w-1$, $d=r+g-1$, $\omega_1=r$, and $\omega_2=g$. Construct a comparison recursion on $N$ check-input density pairs to be updated. Extend the profile by $x_j=\gd$ for $j\leq0$ and $x_j=x_N$ for $j>N$. Form $y_i=w^{-1}\sum_{k=0}^{w-1}\mathsf g(x_{i+k})$, take $h_i=\mathsf f_c(y_i)$ for every $i\geq1$ and $h_i=\gd$ for $i\leq0$, and update $x_1,\ldots,x_N$ by~\eqref{eq:coupled}. After each update reset the constant right extension to the new $x_N$.

The perfect left boundary starts an ordered profile, and the constant right extension permits its sums to telescope under a shift. This comparison removes information available from the actual right shortening. Proving convergence for the resulting less informative recursion therefore suffices for the terminated system.

Start these pairs at $(\Delta_0,\Delta_0)$. The map is monotone by~\cite[Proposition 3(i)]{KumarYoungMacrisPfister2014}, and this initialization is an upper bound in the degradation order. Its iterates therefore become more informative and converge by~\cite[Proposition 11(v)]{KumarYoungMacrisPfister2014}. They increase in degradation with spatial index: forward averaging, $\mathsf g$, $\mathsf f_c$, and backward averaging preserve this order, as do the perfect left boundary and constant right extension. These orders persist at the limit. The comparison system dominates the finite terminated recursion by induction: it replaces right exterior shortening by variables with channel observations and a right extension that may not have perfect information, and contains all active variables of the terminated system.

Suppose its limiting fixed profile $x$ does not have perfect information at every position, and put $p=x_N$. Then $p\neq\gd$ by spatial order. Every check input is at least as informative as $p$, so every $h_i\preceq\DE(p)$, including shortened left values. At the last updated position,
\begin{equation}\label{eq:plateau}
p=\frac1w\sum_{k=0}^{w-1}h_{N-k}\preceq\DE(p),\qquad U(p;c)\geq\gap(c),
\end{equation}
where the last inequality is Lemma~\ref{lem:descent}.

After constructing the limiting profile, we study a finite functional with prescribed exterior data. At the limiting profile the prescribed value agrees with $x_N$, so the interior update is unchanged. Keeping that prescribed value fixed during differentiation ensures that the derivative contains exactly the updates of the $N$ variables of this functional.

When varying the density pairs of the profile $z$, hold the exterior right density $p$ fixed, even while varying $z_N$. Extend $z$ by $\gd$ to the left and this fixed $p$ to the right, and define the finite functional
\begin{equation}\label{eq:finiteU}
\mathcal U_{N,p}(z;c)=\sum_{j=1}^NQ(z_j)-\sum_{i=1}^NF_c\left(\frac1w\sum_{k=0}^{w-1}\mathsf g(z_{i+k})\right).
\end{equation}
Let $(\mathsf Sx)_j=x_{j-1}$ at the updated positions, with the exterior unchanged. Since the original constant right extension equals $x_N=p$, shifting the full profile agrees with this definition. Thus the shifted forward average at $i$ is the original $y_{i-1}$, including at $i=N$, where $y_N=\mathsf g(p)$. Both sums telescope to give the exact identity
\begin{align}
\mathcal U_{N,p}(\mathsf Sx;c)-\mathcal U_{N,p}(x;c)
&=-Q(p)-F_c(y_0)+F_c(\mathsf g(p))\nonumber\\
&=-U(p;c)-F_c(y_0)\leq-\gap(c).\label{eq:shift}
\end{align}
The first equality follows by cancellation of the intermediate terms in~\eqref{eq:finiteU}; the second uses~\eqref{eq:QF} and~\eqref{eq:U}; the inequality uses~\eqref{eq:plateau} and $F_c\geq0$. This is the two-type counterpart of the shift calculation in~\cite[Lemma 41 and Appendix IV-A]{KumarYoungMacrisPfister2014}, with the left boundary term $F_c(y_0)$ retained.

Write $\mathsf T_p$ for the update of these density pairs with this fixed exterior. Use the convolution differentiation rules~\cite[Propositions 14--15]{KumarYoungMacrisPfister2014} and signed entropy duality~\cite[Proposition 5]{KumarYoungMacrisPfister2014}. The same cancellation as in~\eqref{eq:dU}, retaining each averaging weight in~\eqref{eq:finiteU}, gives at every profile $z$ on these positions
\begin{equation}\label{eq:firstvariation}
D\mathcal U_{N,p}(z)[v]
=\sum_{j=1}^N\sum_{\alpha=1}^2\omega_\alpha
H\left((z_{j,\alpha}-\mathsf T_{p,j,\alpha}(z))\star D\mathsf g_\alpha(z_j)[v_j]\right).
\end{equation}
For an updated position $j$, all contributing variable indices are $i\in[j-w+1,j]\cap[1,N]$, exactly those retained in~\eqref{eq:finiteU}. The omitted $i\leq0$ correspond to $h_i=\gd$, whose convolution with any mass-zero variation has zero entropy. The fixed right exterior has zero variation. Thus neither boundary contributes an omitted derivative, and $\mathsf T_p(x)=x$.

Take $v=\mathsf Sx-x$. With $h_i=\gd$ for $i\leq0$, subtracting the two neighboring fixed-point averages in~\eqref{eq:coupled} cancels their $w-1$ common terms and gives, also for $j=1$,
\begin{equation}\label{eq:increment}
v_j=x_{j-1}-x_j=\frac{h_{j-w}-h_j}{w}.
\end{equation}
For each type $\alpha$, put $e_{j,\alpha}=H(x_{j,\alpha})-H(x_{j-1,\alpha})\geq0$. Summing this definition cancels the interior entropies, and $x_0=\gd$, $x_N=p$, and $H(p_\alpha)\leq1$ give
\begin{equation}\label{eq:telescope}
\sum_{j=1}^N e_{j,\alpha}=H(p_\alpha)\leq1.
\end{equation}

These two bounds control different factors in the variation estimate. Equation~\eqref{eq:increment} provides an explicit factor $1/w$. Equation~\eqref{eq:telescope} bounds the sum of the neighboring entropy differences by one, independently of the chain length. Retaining this sum, rather than bounding each difference separately by one, is what removes any factor proportional to $N$.
Along $z_s=x+sv$, write $G_{j,\alpha}(s)=D\mathsf g_\alpha(z_{s,j})[v_j]$. Each component of $v$ is an informative-minus-degraded difference, and $G_{j,\alpha}(s)$ is a sum of $d$ differences with this orientation. The signed-entropy inequality~\cite[Proposition 8(iii)]{KumarYoungMacrisPfister2014} therefore gives $H(v_{j,\alpha}\star G_{j,\alpha}(s))\geq0$. Integrating the directional derivative of $\mathsf T_p$ along the mixture path and using $\mathsf T_p(x)=x$ expresses the residual as
\[
z_{s,j,\alpha}-\mathsf T_{p,j,\alpha}(z_s)
=s v_{j,\alpha}-\int_0^sD\mathsf T_{p,j,\alpha}(z_\sigma)[v]\,d\sigma.
\]
Substitution into~\eqref{eq:firstvariation} and the nonnegative contribution from $sv$ yield
\[
\frac d{ds}\mathcal U_{N,p}(z_s)
\geq-\int_0^s\sum_{j=1}^N\sum_{\alpha=1}^2\omega_\alpha
H\left(D\mathsf T_{p,j,\alpha}(z_\sigma)[v]\star G_{j,\alpha}(s)\right)\,d\sigma.
\]

By the convolution differentiation rules~\cite[Propositions 14--15]{KumarYoungMacrisPfister2014}, the derivatives select occurrences of input messages in the convolution expressions. Lemma~\ref{lem:tree} controls a selected pair of occurrences using the entropy difference at one of them, regardless of the probability densities at the other leaves. It remains to count these occurrences and use the two bounds above.

To bound this integrand, expand all averages into convex combinations of binary trees of variable and check convolutions. Each nonconstant tree in component $\alpha$ of $\mathsf T_p$ has $d_1=\ell-1$ or $d_2=g-1$ copies of a $d$-leaf check update. Its derivative selects one of $dd_\alpha$ message-leaf occurrences; the channel leaf is fixed. The factor $G_{j,\alpha}(s)$ expands into $d$ terms, each varying one leaf occurrence. For each selected pair, retain the leaf in $G$ as $v_{j,\beta}=x_{j-1,\beta}-x_{j,\beta}$ and use~\eqref{eq:increment} at the leaf in $D\mathsf T_p$. Lemma~\ref{lem:tree} bounds the resulting absolute entropy by $e_{j,\beta}/w$. Unvaried leaves may come from $z_\sigma$ or $z_s$; all are probability densities, so the same lemma applies. Exterior leaves have zero variation. The numbers of type-1 and type-2 leaves in $\mathsf g_1$ are $(r-1,g)$, and those in $\mathsf g_2$ are $(r,g-1)$. Since the averaging weights sum to one, summing over $j$ gives the following bounds for the two components, respectively, by~\eqref{eq:telescope}:
\[
\begin{aligned}
\frac{dd_1}{w}\sum_{j=1}^N\bigl[(r-1)e_{j,1}+g e_{j,2}\bigr]
&\leq\frac{d^2d_1}{w},\\
\frac{dd_2}{w}\sum_{j=1}^N\bigl[r e_{j,1}+(g-1)e_{j,2}\bigr]
&\leq\frac{d^2d_2}{w}.
\end{aligned}
\]
Thus the $d^2d_\alpha$ leaf pairs contribute at most $d^2d_\alpha/w$ in total, with no factor $N$:
\[
\left|\sum_{j=1}^N
H\left(D\mathsf T_{p,j,\alpha}(z_\sigma)[v]\star G_{j,\alpha}(s)\right)\right|
\leq\frac{d^2d_\alpha}{w}.
\]
Consequently,
\begin{equation}\label{eq:shiftderivative}
\frac d{ds}\mathcal U_{N,p}(z_s)
\geq-\frac{d^2(\omega_1d_1+\omega_2d_2)s}{w}
=-\frac{\mathcal K_{\ell,r,g}s}{w}.
\end{equation}
As in~\cite[Proposition 16]{KumarYoungMacrisPfister2014}, multilinearity makes the mixture-path functional a finite polynomial. Integrating~\eqref{eq:shiftderivative} from zero to one gives
\[
\mathcal U_{N,p}(\mathsf Sx;c)-\mathcal U_{N,p}(x;c)
\geq-\frac{\mathcal K_{\ell,r,g}}{2w}.
\]
Together with~\eqref{eq:shift} this contradicts~\eqref{eq:widthgeneral}. The comparison limit has perfect information, as does the terminated limit below it in degradation order. This proves convergence from the least informative initialization and hence also from the initialization from channel observations. Perfect check-input limits imply perfect bit posteriors.
\end{proof}

\begin{corollary}[Sufficient coupling width at degree three]\label{prop:widthmain}
Set $r=g=3$, fix an integer $\ell\geq4$, and let $R=r/\ell$. For every BMS density $c$ with $C(c)=1-H(c)>R$, the uncoupled potential defined in~\eqref{eq:U} has a positive minimum $\gap(c)$ over its density-evolution fixed points other than $\gd$. For every finite active length $L$, the terminated coupled recursion converges to perfect information if
\begin{equation}\label{eq:widthmain}
w>\frac{\mathcal K_{\ell,r,g}}{2\gap(c)}.
\end{equation}
\end{corollary}
\begin{proof}
Combine Theorems~\ref{thm:gap} and~\ref{thm:saturation}.
\end{proof}
The density-evolution statement takes the section size to infinity before the number of iterations, for fixed $L,w$. The code-sequence statement in Theorem~\ref{thm:main} uses a diagonal choice of these parameters with $L/w\to\infty$, as proved in Section~\ref{sec:rate}.

\section{A degree-two obstruction on symmetric channels}\label{sec:degreetwo}
The degree-two case tests the remaining step of the BEC strategy: positivity on $\Fnt(c)$ from~\eqref{eq:fixedpoint-sets}. In the first two results set $r=g=2$, $\ell>r$, and $R=r/\ell$. The two trivial fixed points still have potentials $0$ and $C(c)-R$, but we prove that, for each $\ell\in\{3,4,5\}$, some BSCs with $C(c)>R$ have an uncoupled fixed point with negative potential. Such a fixed point $p$ belongs to $\Fnt(c)$, since neither trivial point has negative potential above the rate.

A conditional obstruction, stated for general degrees below, then shows that their terminated density evolution retains positive average transmitted-bit error when $L/w\to\infty$. This concerns the density-evolution limit of the ensemble in Section~\ref{sec:model}; it is not a MAP converse or an impossibility statement for every finite-code realization.

A trial pair with negative potential is useful only if it can be related to the initialized density-evolution trajectory. We first minimize a one-density functional whose minimizer satisfies the punctured-variable fixed-point equation. The full update then improves the transmitted coordinate, giving an ordered iteration along which $U$ decreases. Its limit is a fixed point of both equations with negative potential.

\begin{lemma}[A density minimizer and monotone iteration]\label{lem:scalarminimum}
Let $r=g=2$, $\ell>r$, $R=r/\ell$, and $I(c)>0$. With the density of messages from transmitted variables fixed at $b=c$, let $a$ range over all BMS densities of messages from punctured variables. Put $W=c^{\cc g}$ and $q=a\cc W$, and define
\begin{equation}\label{eq:Vdefinition}
V(a;c)=-(r-1)H(a^{\cc r}\cc W)+rH(q)-R H(q^{\star\ell}).
\end{equation}
Every global minimizer $a_*$ satisfies $a_*=(a_*\cc W)^{\star(\ell-1)}$. Moreover,
\begin{equation}\label{eq:VupperU}
U(a,c;c)\le V(a;c)-H(c).
\end{equation}
If $V(a_0;c)<H(c)$ for some trial density $a_0$, then $\DE$ has a fixed point $p$ other than $\gd$ with $U(p;c)<0$.
\end{lemma}
\begin{proof}
Compactness of reliability laws on $[0,1]$ and continuity give a minimizer~\cite[Proposition 11(ii)--(iv)]{KumarYoungMacrisPfister2014}. We use the minimum-direction method of~\cite[Lemma 24]{KumarYoungMacrisPfister2014}, deriving the required identity for $V$. Write $a'=(a\cc W)^{\star(\ell-1)}$. Differentiate each convolution power by~\cite[Propositions 14--15]{KumarYoungMacrisPfister2014}, then apply signed entropy duality~\cite[Proposition 5]{KumarYoungMacrisPfister2014}. Since $r-1=1$ and $R\ell=r$, the resulting terms combine to give
\[
dV=rH((a-a')\star(da\cc W)).
\]
The mixture $(1-s)a+sa'$ remains a BMS density for $0\leq s\leq1$, so its right derivative at a minimum must be nonnegative. This argument also applies to a minimizer on the boundary of the space of densities. The expression below is nonpositive and can therefore vanish only when all the squared moment differences vanish.

At a minimizer take the admissible direction $da=a'-a$ and set $\nu=a-a'$. Signed duality first replaces the variable convolution by a check convolution; the moment product rule and signed entropy expansion~\cite[Propositions 5, 6(iii), 7, and 8(ii)]{KumarYoungMacrisPfister2014} then give
\[
dV=-r\sum_{k\ge1}\gamma_k M_k(\nu)^2M_k(W).
\]
Since $I(c)>0$, the absolute channel reliability is positive with positive probability. Hence $M_k(c)>0$ for every $k\geq1$, and the check-convolution product rule~\eqref{eq:moments} gives $M_k(W)=M_k(c)^g>0$. A minimum therefore has $M_k(\nu)=0$ for every $k$, which implies $\nu=0$ by uniqueness of moments on $[0,1]$~\cite[Proposition 11(i) and Appendix I]{KumarYoungMacrisPfister2014}.

For $v=a^{\cc r}\cc c^{\cc(g-1)}$, substitute $b=c$ in~\eqref{eq:U}, subtract the definition of $V$, and use~\eqref{eq:duality}~\cite[Proposition 4]{KumarYoungMacrisPfister2014}. This gives
\[
U(a,c;c)=V(a;c)-H(c)
+I(c\star v^{\star g})-gI(c\star v)+(g-1)I(c).
\]
Here $g=2$. Observe a uniform bit $X$ through channels $Y\sim c$ and $Z_1,Z_2\sim v$ that are independent conditional on $X$. The chain rule identifies the final three terms with $-I(Z_1;Z_2\mid Y)$, and nonnegativity of conditional mutual information gives the inequality $-I(Z_1;Z_2\mid Y)\le0$~\cite[Chapter 2]{CoverThomas2006}.

Although $Z_1$ and $Z_2$ are independent given $X$, they can share information about $X$ after only $Y$ is known. Their conditional mutual information measures this shared part. Subtracting it explains why the auxiliary functional gives an upper bound for $U$. This proves~\eqref{eq:VupperU}. At a minimizer, $x=(a_*,c)$ satisfies $\DE(x)\preceq x$, because its punctured-variable message density is unchanged by the update and its transmitted coordinate adds observations to $c$. Convergence of the monotone iteration and continuity of the update~\cite[Proposition 11(iv)--(v)]{KumarYoungMacrisPfister2014} give a fixed-point limit $p$. In the following chain, the first inequality uses Lemma~\ref{lem:descent}; the second uses~\eqref{eq:VupperU} and minimality of $a_*$:
\[
U(p;c)\le U(a_*,c;c)\le V(a_0;c)-H(c)<0.
\]
The limit cannot be $\gd$, whose potential is zero.
\end{proof}

\begin{proposition}[Fixed points with negative potential above the design rate]\label{prop:negativefp}
Set $r=g=2$ and $R=r/\ell$. For each $\ell\in\{3,4,5\}$, there is a nonempty interval of BSCs with $C(c)>R$ and $\gap(c)<0$.
\end{proposition}
\begin{proof}
The local condition concerns bifurcation from the trivial fixed point $(\Delta_0,c)$, called the paramagnetic solution in the statistical-mechanical analysis~\cite{TanakaSaad2003}. For $r=2$, the condition in that analysis is $\int v\rho(v)\,dv=(C-1)^{-1/L}$, where $v$ is the signed channel reliability and $\rho$ its law conditional on the reference transmitted bit. The BMS reliability representation in Section~\ref{sec:preliminaries} gives $\E v=\E v^2=M_1(c)$~\cite[Section II-A]{KumarYoungMacrisPfister2014}, so $(K,C,L)=(r,\ell,g)$ converts the boundary to $(\ell-1)M_1(c)^g=1$. We derive its strict potential consequence:
\begin{equation}\label{eq:treeinstability}
(\ell-1)M_1(c)^g>1
\quad\Longrightarrow\quad
\min_a V(a;c)<1-R.
\end{equation}
Take $a_\eta=(1-\eta)\Delta_0+\eta d$. The first line below uses the binomial expansion of convolution powers~\cite[Proposition 14]{KumarYoungMacrisPfister2014} and duality~\eqref{eq:duality}; the second uses the entropy series and moment product rule in~\eqref{eq:moments}~\cite[Propositions 6(iii) and 7]{KumarYoungMacrisPfister2014}:
\begin{align*}
V(a_\eta;c)
&=1-R+\eta^2\{I(d^{\cc2}\cc W)
 -(\ell-1)I(d^{\cc2}\cc W^{\cc2})\}+O(\eta^3)\\
&=1-R+\eta^2\sum_{k\ge1}\gamma_kM_k(d)^2M_k(c)^g
 \{1-(\ell-1)M_k(c)^g\}+O(\eta^3).
\end{align*}
The expansion is a finite polynomial in the mixture parameter $\eta$. To verify its quadratic term, the punctured-bit posterior has coefficient
$\binom\ell2 I((d\cc W)^{\cc2})$ in its entropy expansion, and $R\binom\ell2=\ell-1$ under $r=2$. Choose $d$ to be a BSC of absolute reliability $\varepsilon$. Then $M_k(d)^2=\varepsilon^{4k}$, so the sign of the coefficient for sufficiently small $\varepsilon>0$ is the sign of $1-(\ell-1)M_1(c)^g$. Taking sufficiently small $\eta>0$ proves~\eqref{eq:treeinstability}.

The choices have a definite order: first fix $\varepsilon$ so that the quadratic coefficient is strictly negative, and then decrease $\eta$ so that the higher-order terms cannot change the sign. The first choice is valid because the $k=1$ term has order $\varepsilon^4$, whereas the sum over $k\geq2$ is $O(\varepsilon^8)$; the weights sum to one and the remaining coefficients are bounded for fixed $\ell$. The second choice uses the finite polynomial in $\eta$.

For a BSC write $n=M_1(c)$, so $C(c)=f(n)$ by the BSC definition of $f$ preceding~\eqref{eq:fbound}. The positive logarithm series~\cite[Eq. (4.6.4)]{NISTDLMF} gives $\ln2>842/1215$ from its first three terms at $1/3$. Truncating the moment series~\eqref{eq:moments}~\cite[Proposition 7]{KumarYoungMacrisPfister2014} and bounding its tail by a geometric series therefore gives, for $0<n<1$,
\[
f(n)<B(n):=\frac{1215}{842}
\left\{\sum_{k=1}^{4}\frac{n^k}{2k(2k-1)}
+\frac{n^5}{90(1-n)}\right\}.
\]
Here $2k(2k-1)\ge90$ for the tail $k\ge5$. Direct substitution into this expression for $B$ gives the following rational comparisons:
\[
\begin{array}{c|c|c|c}
\ell&\bar n&B(\bar n)<R&(\ell-1)\bar n^{g}\\ \hline
3&3/4&7875387/12070912<2/3&9/8\\
4&3/5&14529699/29470000<1/2&27/25\\
5&1/2&301347/754432<2/5&1
\end{array}
\]
Let $c_0$ be the BSC of capacity $R$. Since $f$ is strictly increasing, $M_1(c_0)>\bar n$, and hence~\eqref{eq:treeinstability} holds strictly in all three cases.

In particular, the equality in the last row of the table causes no loss of strictness: that row uses the lower comparison value $\bar n$, whereas the actual capacity-boundary BSC has strictly larger moment. Fix a trial density with $V(a_0;c_0)<1-R=H(c_0)$. This strict inequality persists when the BSC is improved slightly, by continuity of entropy and both convolutions~\cite[Proposition 11(iii)--(iv)]{KumarYoungMacrisPfister2014}. These improved channels have capacity greater than $R$, and Lemma~\ref{lem:scalarminimum} gives the claimed fixed points with negative potential. All comparisons above are analytic rational inequalities; no finite computational certificate is used.
\end{proof}

\begin{proposition}[An obstruction to terminated density evolution]\label{prop:coupledconverse}
Let $\ell>r\ge2$, $2\le g\le\ell$, and $R=r/\ell$. Suppose $p$ is an uncoupled fixed point with $u=-U(p;c)>0$ and $z=Z(c)>0$. Use $\overline P_{\mathrm{tx}}(L,w;c)$ as defined in~\eqref{eq:txDEerror}, and write $[t]_+=\max\{t,0\}$. Then
\begin{equation}\label{eq:txobstruction}
\overline P_{\mathrm{tx}}(L,w;c)
\ge\frac{z^2[\,u-(w-1)Q(p)/L\,]_+^2}
{4(2^gR+z)^2}.
\end{equation}
In particular, for every channel supplied by Proposition~\ref{prop:negativefp} and every sequence $L/w\to\infty$,
\[
\liminf\overline P_{\mathrm{tx}}(L,w;c)
\ge\frac{z^2u^2}{4(2^gR+z)^2}>0.
\]
\end{proposition}
\begin{proof}
The proof compares an interior contribution proportional to $L$ with a boundary contribution proportional to $w-1$. Initializing at the negative-potential fixed point makes this comparison explicit, and monotone descent preserves the resulting upper bound on the finite potential. We then compare with the channel initialization and convert the remaining posterior entropy into transmitted-bit error.

Use $Q,F_c$ from~\eqref{eq:QF} and put $N=L+w-1$ and $(\omega_1,\omega_2)=(r,g)$. For check-input density pairs $x_1,\ldots,x_N$, define the finite functional
\[
\mathcal E_{L,w}(x;c)=\sum_{j=1}^{N}Q(x_j)
-\sum_{i=1}^{L}F_c\left(\frac1w\sum_{k=0}^{w-1}\mathsf g(x_{i+k})\right).
\]
Both exterior variable regions are shortened. The cancellation proving~\eqref{eq:firstvariation}, using signed duality and convolution differentiation~\cite[Propositions 5 and 14--15]{KumarYoungMacrisPfister2014}, now retains only active variable sections $1,\ldots,L$; omitted shortened outputs have zero signed entropy. Thus, for the actual terminated update $\mathsf T_{L,w}$,
\[
D\mathcal E_{L,w}(x)[v]
=\sum_{j,\alpha}\omega_\alpha
H((x_{j,\alpha}-\mathsf T_{L,w,j,\alpha}(x))
\star D\mathsf g_\alpha(x_j)[v_j]).
\]
Consequently $\mathcal E_{L,w}$ decreases along iterations monotone in the degradation order, by the proof of Lemma~\ref{lem:descent}.

Substituting the entropy series and moment product rule~\eqref{eq:moments}~\cite[Propositions 6(iii) and 7]{KumarYoungMacrisPfister2014} into~\eqref{eq:QF} shows $0\le Q(a,b)\le1$: its $k$th kernel, weighted by $\gamma_k$, is
\[
1-rA_k^{r-1}B_k^g-gA_k^rB_k^{g-1}
+(r+g-1)A_k^rB_k^g,
\]
the probability of at least two failures among $r$ independent trials of success probability $A_k$ and $g$ of success probability $B_k$.

For comparison, initialize every check-input density pair to the fixed point $p$. Equation~\eqref{eq:coupled} replaces it after one update at position $j$ by $\lambda_jp+(1-\lambda_j)\gd\preceq p$, where
$\lambda_j=|[j-w+1,j]\cap[1,L]|/w$. The iteration is monotone in the degradation order and converges to a density profile $x^*$~\cite[Proposition 11(v)]{KumarYoungMacrisPfister2014}. Descent gives the following inequality, and $U=Q-F_c\circ\mathsf g$ from~\eqref{eq:QF} gives the equality:
\[
\mathcal E_{L,w}(x^*;c)
\le NQ(p)-LF_c(\mathsf g(p))=-Lu+(w-1)Q(p).
\]

There are $L+w-1$ check sections and $L$ active variable sections, so their difference leaves exactly $(w-1)Q(p)$ in this expression. Since $0\leq Q(p)\leq1$, the positive boundary contribution per active section vanishes when $L/w\to\infty$. The negative interior contribution per section remains $-u$.
Initializing the variable messages with channel observations gives $u_0=(\Delta_0,c)\succeq p$. The resulting initial check-input density profile is $\lambda_ju_0+(1-\lambda_j)\gd$, which dominates, in the degradation order, the first update from the fixed-point initialization used for comparison. Monotonicity therefore makes the limiting posteriors of the iteration initialized with channel observations no more informative than those of $x^*$. Write $y_i=\lim_{t\to\infty}y_{i,t}$ for this iteration in Section~\ref{sec:model}. Since $Q\ge0$, their average posterior entropy contribution satisfies
\begin{equation}\label{eq:posteriorobstruction}
\frac1L\sum_{i=1}^{L}F_c(y_i)
\ge u-\frac{w-1}{L}Q(p).
\end{equation}

The quantity $F_c$ includes both punctured-bit and transmitted-bit posterior entropies. A positive lower bound on their sum must still be connected to the transmitted bits. The common check inputs allow us to bound the punctured contribution by the transmitted Bhattacharyya parameter, after which a standard entropy-to-error estimate gives the desired conclusion.

To obtain transmitted error alone, put $q_0=a^{\cc(r-1)}\cc b^{\cc(g-1)}$, so $q_1=q_0\cc b$ and $q_2=q_0\cc a$. Since $g\ge2$, degradation under check convolution~\cite[Proposition 3]{KumarYoungMacrisPfister2014} and~\eqref{eq:HZ} give
$Z(q_1)\le Z(q_0)+Z(b)\le2Z(q_2)$. The factor $2$ counts these two bounds. Thus, with $v_i=Z(y_{i,2})$, linearity under mixtures gives $Z(y_{i,1})\le2v_i$. Apply $H\leq Z$ and multiplicativity of $Z$ in~\eqref{eq:HZ}~\cite[Section II-D and Appendix II-D]{KumarYoungMacrisPfister2014}, then use $\ell\ge g$:
\[
H(y_{i,1}^{\star\ell})\le\min\{1,(2v_i)^\ell\}
\le2^g v_i^g,
\qquad Z(c\star y_{i,2}^{\star g})=zv_i^g.
\]
Substitution into the definition of $F_c$ in~\eqref{eq:QF} gives $F_c(y_i)\le(2^gR+z)Z(c\star y_{i,2}^{\star g})/z$.
For any BMS posterior $d$, apply Cauchy--Schwarz to the reliability representation of $Z(d)$ in Section~\ref{sec:preliminaries} and use~\eqref{eq:Pe} to obtain
$Z(d)\le2\sqrt{P_e(d)(1-P_e(d))}\le2\sqrt{P_e(d)}$.
Averaging and applying Jensen's inequality for the concave square root~\cite[Section 3.1.8]{BoydVandenberghe2004} to~\eqref{eq:posteriorobstruction} proves~\eqref{eq:txobstruction}.
\end{proof}

\begin{proof}[Completion of Theorem~\ref{thm:degreetwo}]
Proposition~\ref{prop:negativefp} supplies the interval of channels and a fixed point with negative potential for each channel. Proposition~\ref{prop:coupledconverse} then gives the claimed strictly positive limiting error bound.
\end{proof}

\section{Actual rate and operational limits}\label{sec:rate}
This section completes the capacity proof by selecting code realizations with vanishing average BP bit error and proving that their actual transmitted rate tends to $R$. The trivial fixed-point value $C(c)-R$ identifies the candidate capacity boundary; the decoding result below controls the dimension lost under puncturing, and an auxiliary BEC gives a matching upper bound on the transmitted rate.

Let $\ell>r\geq2$ and $g\geq2$. The passage from density evolution to ensemble-average performance uses the computation-tree argument~\cite[Section IV, Theorem 2]{RichardsonUrbanke2001}. We apply that argument to the balanced two-type construction in Section~\ref{sec:model}. For fixed degrees, $L,w$, and finite iteration count $t_{\mathrm{iter}}$, exploring a bit's computation neighborhood reveals only a bounded number of sockets. Every matching group has size proportional to $M$. The probability of revisiting an exposed node is therefore $O(M^{-1})$, and sampling the bounded number of offsets without replacement differs from independent uniform sampling by $O(M^{-1})$. Thus the finite neighborhood and the independent computation tree defining density evolution can be coupled to agree except on an event of probability $O(M^{-1})$. The constants may depend on the fixed degrees, $L,w$, and $t_{\mathrm{iter}}$. Since a bit-error indicator is bounded by one, the ensemble-average BP bit error probabilities for both variable types converge to their density-evolution values. The bound is needed only at fixed iteration count; after taking this section-size limit, let the iteration count increase. When $\gap(c)>0$ and the width satisfies~\eqref{eq:widthgeneral}, Theorem~\ref{thm:saturation} gives zero limiting errors of both types.

The rate after puncturing requires a separate argument. For a fixed realization let $[A\ B]$ be its extended parity-check matrix, with punctured coordinates in $A$ and transmitted coordinates in $B$. All ranks and dimensions in this section are over $\mathbb F_2$. With $n=N_{\mathrm{tx}}$, the dimensions of the projected and extended codes satisfy
\begin{equation}\label{eq:rank}
k_{\mathrm{tx}}=n-\rank[A\ B]+\rank A
=k_{\mathrm{ext}}-\dim\ker A.
\end{equation}
For any fixed transmitted vector $v$, two compatible punctured vectors $u,u'$ satisfy $A(u-u')=0$. Thus each nonempty set of compatible punctured vectors is a coset of $\ker A$ and contains $2^{\dim\ker A}$ vectors. Projection identifies exactly this many extended codewords, which accounts for the dimension subtracted in~\eqref{eq:rank}.

The first identity follows by projecting the extended kernel onto transmitted coordinates; its projection kernel is exactly $\ker A$. Using the variable counts in Table~\ref{tab:degrees_model} and bounding the rank by the number $(L+w-1)M$ of retained checks in Section~\ref{sec:model}, including redundant or zero checks, gives
\begin{equation}\label{eq:rawrate}
k_{\mathrm{ext}}\geq n\left(R-\frac{w-1}{L}\right).
\end{equation}
No rank assumption has been made.

\begin{proposition}[Actual rate under positive fixed-point gaps]\label{prop:generalrate}
For general $\ell>r\geq2$ and $g\geq2$, fix a BMS channel $c$ with $C(c)>R=r/\ell$. Suppose $\gap(c)>0$ and $\gap(\mathrm{BEC}(\epsilon))>0$ for every $0<\epsilon<1-R$. Then the terminated, uniformly smoothed MN ensembles admit code realizations with actual transmitted rate tending to $R$ and vanishing average BP bit error probabilities for punctured and transmitted bits on $c$.
\end{proposition}
\begin{proof}
An auxiliary BEC whose capacity approaches $R$ serves two purposes: reliability of the punctured bits bounds the dimension lost under puncturing, $\dim\ker A$, to give a lower rate bound, while transmitted-bit reliability gives a matching upper bound. We select the same realizations to be reliable on $c$.
Take BEC erasure probabilities
\[
\epsilon_j=(1-R)\left(1-\frac1{j+1}\right)\uparrow1-R.
\]
Choose $w_j\geq j$ satisfying~\eqref{eq:widthgeneral} for both $c$ and $\mathrm{BEC}(\epsilon_j)$, and then $L_j$ with $w_j/L_j\leq1/j$. On these fixed finite chains choose a common finite iteration count large enough to make all four density-evolution bit error probabilities (two variable types on two channels) smaller than a chosen sequence tending to zero. Choose an admissible $M_j$ large enough for the finite-iteration computation-tree approximation on both channels, using the fixed-iteration convergence discussed above~\cite[Theorem 2]{RichardsonUrbanke2001}. More explicitly, if $\mathcal E_j$ denotes the sum of the four finite-iteration bit error probabilities of a random realization and $\eta_j=1/(j+1)$, the choices can ensure $\E\mathcal E_j\leq\eta_j^2$. Markov's inequality gives $\Pr\{\mathcal E_j>\eta_j\}\leq\eta_j<1$. Hence one realization at each $j$ has $\mathcal E_j\leq\eta_j$, so all four errors tend to zero. Write the two BEC errors as $e_{u,j}$ and $e_{v,j}$.

The finite list of four errors is what permits selection of one realization: their nonnegative sum has a small expectation, so at least one code makes the whole sum small. Choosing separate realizations for the target channel and the auxiliary BEC would not establish the rate of the code used on the target channel.

For a uniform extended codeword $(U,V)$, $V$ is uniform on the projected code, since a linear projection has equal-size fibers. The codeword-symmetry argument~\cite[Lemma 1]{RichardsonUrbanke2001}, with symmetric tie breaking, identifies the bit error probabilities used above with those under this uniform distribution. Conditional on $V$, the vector of punctured bits is uniform on a coset of $\ker A$, so
\[
\dim\ker A=H(U\mid V)\leq H(U\mid Y_j)
\leq n_u h_2(e_{u,j})=o(n),\qquad n_u=Rn.
\]
Here $Y_j$ is the BEC observation. The equality uses the coset count underlying~\eqref{eq:rank}. The first inequality is data processing for the Markov chain from $U$ through $V$ to $Y_j$. For the second, the entropy chain rule bounds joint conditional entropy by the sum of coordinatewise conditional entropies; binary Fano bounds each summand, and concavity of $h_2$ replaces the individual errors by their average~\cite[Chapter 2]{CoverThomas2006}. Equations~\eqref{eq:rank}--\eqref{eq:rawrate} yield $k_{\mathrm{tx}}/n\geq R-o(1)$.

The coordinatewise entropy bound suffices because the target is a normalized dimension. Vanishing average bit error implies $H(U\mid Y_j)/n\to0$, even without a block-error statement. Hence the ambiguity in the punctured bits after revealing all transmitted bits also has sublinear dimension. Conversely, the memoryless-channel information bound and the same coordinatewise Fano bound give
\[
k_{\mathrm{tx}}=H(V)=I(V;Y_j)+H(V\mid Y_j)
\leq n(1-\epsilon_j)+n h_2(e_{v,j})=n(R+o(1)).
\]
Here the first equality uses uniformity of $V$, and the second is the entropy--mutual-information identity. In the inequality, memorylessness bounds the mutual information by the sum of the single-use capacities, even when the transmitted coordinates are dependent; the conditional-entropy term uses the coordinatewise Fano bound just derived~\cite[Chapters 2 and 7]{CoverThomas2006}. The last equality uses the definition of $\epsilon_j$ and $e_{v,j}\to0$. Thus the actual rate tends to $R$ along the same sequence of realizations for which the BP bit error probabilities on $c$ tend to zero. This uses simultaneous selection for two channels at each $j$, without a degradation comparison between $c$ and the auxiliary BEC.
\end{proof}

\begin{proof}[Completion of Theorem~\ref{thm:main}]
Set $r=g=3$. Theorems~\ref{thm:gap} and~\ref{thm:saturation} prove the density-evolution assertion, and Corollary~\ref{cor:uniform} proves the common-width assertion. Theorem~\ref{thm:gap} also verifies the channel hypotheses of Proposition~\ref{prop:generalrate}, which supplies the claimed code realizations and actual rate.
\end{proof}

The choice order for a finite realization is $w_j$, then $L_j\gg w_j$, then a sufficiently large iteration count, then $M_j$. This is consistent with taking $M\to\infty$ first in the mathematical definition of density evolution. Taking $L\to\infty$ at a fixed iteration count would not allow the boundary information to reach most sections.

For any sequence of projected codes with vanishing average transmitted bit error $e$, the same conditional-entropy argument, using the memoryless converse and coordinatewise Fano bounds, gives $k_{\mathrm{tx}}\leq nC+nh_2(e)$ on a BMS channel of capacity $C$~\cite[Chapters 2 and 7]{CoverThomas2006}. Therefore reliable actual rate $R$ requires $R\leq C$. Together with Theorem~\ref{thm:main}, this identifies $C=R$ as the attainable boundary. For a continuous complete physically degraded BMS family indexed by entropy, it yields limiting entropy threshold $1-R$. Reliability exactly at the boundary is not needed for this threshold statement.

\section{Discussion}
Comparing the two degree regimes identifies the part of the BEC capacity argument that depends on the channel. The two trivial fixed points have the same potential values, $0$ and $C(c)-R$, in both regimes. For $r=g=3$, positivity at every $x\in\Fnt(c)$ from~\eqref{eq:fixedpoint-sets} completes the comparison above rate; adding the other trivial fixed point and using compactness of $\Fneq(c)$ gives the positive gap used for threshold saturation and the actual-rate conclusion of Theorem~\ref{thm:main}. For $r=g=2$ and $\ell\in\{3,4,5\}$, a nontrivial fixed point with negative potential on some BSCs violates this positivity condition and yields the nonzero limiting transmitted-DE error in Theorem~\ref{thm:degreetwo}, despite capacity achievement on the BEC. The distinction agrees with the earlier prediction for the uncoupled system~\cite{TanakaSaad2003}; the proofs here establish the stated coupled decoding conclusions directly.

The common reductions, conditional saturation theorem, and conditional lower bound on transmitted-bit error retain their general degree assumptions. For $r=g=2$ and $\ell\geq6$, the present arguments do not settle capacity achievement. For degree three, some finite rectangles use the positive term $(I(b)-R)D$ in addition to a lower bound on $S$. The auxiliary assertion $S(a,b)\geq0$ for arbitrary BMS density pairs remains open in this regime, whereas the required fixed-point positivity is proved. Further analytic estimates could shorten the finite verification.

The capacity theorem fixes the degree parameters while increasing coupling width, chain length, and section size; the degree-two theorem applies whenever $L/w\to\infty$, including growing widths. The common-width result away from capacity does not supply an explicit numerical gap. Vanishing block error, arbitrary edge spreading, and a single explicit code sequence reliable on an unordered class of BMS channels require further arguments.

\subsection{Future research}\label{sec:future}
The degree-three capacity theorem supplies the starting point for the applications below. Applications based on this proof require the assumptions that give a positive fixed-point gap, or a corresponding sign condition for a modified recursion; the degree-two result shows why the BEC correspondence alone does not suffice.

\paragraph{Rate-compatible and rateless transmission.}
Random puncturing gives a direct connection between Theorem~\ref{thm:main} and incremental redundancy. Fix a mother ensemble with $r=g=3$ and design rate $R_0=r/\ell$. Reveal each of its originally transmitted variables independently with probability $p\in(0,1]$, and puncture the others. If the revealed variables pass through a BMS channel $c$ and the receiver knows their positions, the effective channel density and capacity follow by linearity of entropy under mixtures~\cite[Section II-C]{KumarYoungMacrisPfister2014}:
\[
c_p=(1-p)\Delta_0+pc,\qquad C(c_p)=pC(c).
\]
Applying Corollary~\ref{prop:widthmain} to $c_p$, the density-evolution condition becomes $pC(c)>R_0$. If $N$ is the mother transmitted length and $n_p$ is the number of revealed bits, then $n_p/N\to p$ in probability, so the effective design rate tends to $R_0/p$. Independent uniform marks on the variable positions produce nested revelation sets as $p$ increases. This suggests an incremental-redundancy construction over a prescribed capacity range above $R_0$. An operational result must select one mother-code sequence and transmission order for the required channels and stages, and control block error and stopping decisions. An unlimited output stream requires an additional construction. The BEC result for spatially coupled precoded rateless codes~\cite{SakataKasaiSakaniwa2014} and the rate-adaptive MN analysis~\cite{ZahrLiva2025} provide relevant precedents; extending the former's potential-duality argument to general BMS densities is a separate question.

\paragraph{Distributed compression and information reconciliation.}
For independent copies of a pair $(X,Y)$ with a uniform binary $X$ and BMS conditional law $P_{Y\mid X}$, the virtual channel capacity is $1-H(X\mid Y)$ by the mutual-information identity and symmetry~\cite[Chapters 2 and 7]{CoverThomas2006}. Syndrome coding~\cite{PradhanRamchandran2003} therefore suggests converting the projected MN code into a Slepian--Wolf code with syndrome rate tending to $1-R$, with the corresponding boundary at $H(X\mid Y)=1-R$. Such a construction must account for the rank after puncturing, preserve a sparse representation for coset decoding, and strengthen average-bit reliability to block reliability. Information reconciliation in quantum key distribution is a related application, provided the induced binary channel satisfies the required assumptions. Rate adaptation, public information leakage, and frame error must be evaluated together, as in existing reconciliation studies~\cite{YangEtAl2024}. The security analysis additionally depends on the quantum protocol.

\paragraph{Common codes and parallel channels.}
The common-width conclusion motivates a finite-length construction that performs well on several BMS channels with matched LLRs. The simultaneous-selection argument of Section~\ref{sec:rate} suggests first treating a finite set of channels; a single explicit construction for an entire unordered class remains a further objective. Parallel channels give another application: when a channel label $J$ is drawn independently at each use, independently of the input, and is revealed to the receiver, BMS components $c_j$ form a BMS channel of capacity $\sum_j\pi_j C(c_j)$, where $\pi_j=\Pr\{J=j\}$. The capacity expression follows from the mutual-information chain rule, independence of $J$ and the input, and optimality of the uniform input for each BMS component~\cite[Chapters 2 and 7]{CoverThomas2006}. Optimizing the allocation of such observations along the coupled chain could reduce termination loss or decoding delay. Bit-mapping studies for spatially coupled codes~\cite{HagerEtAl2015} provide a basis for this investigation. Spatially nonuniform allocations require a corresponding density evolution, and applications to bit-interleaved coded modulation must account for the symmetry and dependence of its bit channels.

\paragraph{Windowed decoding and channels with memory.}
Windowed decoding restricts updates to a moving part of the chain. Extending its BEC analysis~\cite{IyengarEtAl2013} to the present two-type BMS recursion would clarify how the capacity gap and target error determine the window size and iteration count. Fixed degrees alone do not bound the total decoding work as the gap closes. Another direction is to extend the generalized-erasure result with memory~\cite{FukushimaOkazakiKasai2015} to finite-state channels with general noise. This requires a joint detector--decoder recursion, a channel contribution expressed through conditional entropy rates, and a new fixed-point positivity argument. The initial target is the symmetric information rate; identifying it with channel capacity requires an additional property of the channel.

\appendix
\section{Elementary scalar estimates}\label{app:elementary}
This appendix supplies the rational estimates used in the analytic part of the degree-three positivity proof. They justify the scalar endpoint comparisons in Section~\ref{sec:bounds} by rational inequalities with explicit remainder bounds; the fixed-point reductions there explain how these comparisons enter the potential argument.

\paragraph*{Elementary evaluation at $q=441/484$}
For the calculation in Lemma~\ref{lem:weak}, set $r=g=3$ and use the coefficients $\alpha_0,\beta_0,\chi_0$ defined in its proof. The identities
\[
f(q)=1-\frac{\ln(44/43)+\ln(43)/44}{\ln2},\qquad
f'(q)=\frac{11\ln43}{42\ln2}
\]
follow by substituting $\sqrt q=21/22$ into the definition $f(u)=1-h_2((1-\sqrt u)/2)$ in Section~\ref{sec:bounds} and differentiating that definition. The rational bounds
\[
\frac{6931}{10^4}<\ln2<\frac{6932}{10^4},\quad
\frac{3760}{10^3}<\ln43<\frac{3763}{10^3},\quad
\frac{229}{10^4}<\ln(44/43)<\frac{230}{10^4}
\]
imply
\[
\frac{843}{10^3}<f(q)<\frac{844}{10^3},\qquad
\frac{1420}{10^3}<f'(q)<\frac{1422}{10^3}.
\]
The logarithm series~\cite[Eq. (4.6.4)]{NISTDLMF} and a geometric bound on its positive tail give, for $0<z<1$,
\[
2\sum_{j=0}^{m-1}\frac{z^{2j+1}}{2j+1}
<\ln\frac{1+z}{1-z}
<2\sum_{j=0}^{m-1}\frac{z^{2j+1}}{2j+1}
 +\frac{2z^{2m+1}}{(2m+1)(1-z^2)}.
\]
The partial sum is a lower bound because every omitted term is positive. For the upper bound, replace each omitted denominator by the first omitted denominator and sum the resulting geometric series. The final inequalities therefore require only finitely many rational operations, including a bound for the entire infinite tail.

Use $(z,m)=(1/3,4),(11/75,2),(1/87,1)$ and
$\ln43=5\ln2+\ln(43/32)$ for the three bounds.
For the function $K$ defined in the proof of Lemma~\ref{lem:weak}, direct substitution into its definition and derivative,
\begin{align*}
K(q)&=(q^{-1}+\alpha_0q^{g-1})f(q)-\beta_0q^{g-2}+\chi_0q^{g-1},\\
K'(q)&=(-q^{-2}+\alpha_0(g-1)q^{g-2})f(q)
 +(q^{-1}+\alpha_0q^{g-1})f'(q)\\
&\quad-\beta_0(g-2)q^{g-3}+\chi_0(g-1)q^{g-2}
\end{align*}
gives $K(q)>7/2000$ and $-1/50<K'(q)<1/50$ using rational arithmetic alone.

\paragraph*{Large-degree comparison}
For the endpoint estimate~\eqref{eq:tail2}, put $x_0=2245489/2640625\allowbreak<1701/2000$.
The four rational inequalities
\[
(1701/2000)^2<579/800,\qquad (579/800)^2<131/250,
\]
\[
(131/250)^2<11/40,\qquad (11/40)^2<5929/78125
\]
give $x_0^{16}<5929/78125$ by successive squaring.

\section{Certificate format and interval arithmetic}\label{app:arithmetic}
This appendix specifies the exact arithmetic and coverage checks that complete the finite part of the nontrivial fixed-point positivity proof. Their role is to justify every sign decision in Section~\ref{sec:certificates}, with bounds valid throughout each closed rectangle. For $r=g=3$ and $C(c)>R$, these checks and the analytic bounds in Section~\ref{sec:bounds} establish positivity on $\Fnt(c)$ from~\eqref{eq:fixedpoint-sets}. Together with the positive value at the other trivial fixed point and compactness, this gives the gap needed for threshold saturation.

The certificates in this appendix use $r=g=3$. The preserved JSON certificate contains the ordered records for $\ell=4,\ldots,46$. Each record specifies its rational rate, dyadic comparison reliability $d_0$, interval for $J(d_0)$, leaf count, depth, rule counts, and leaves. Each leaf stores a binary path, its four rational rectangle endpoints, a rule name, and the rational margin established by that rule. Rule names correspond to the six tests in Section~\ref{sec:certificates}.

\subsection{Coverage and comparisons}
A node is bisected along its longer normalized side. Normalize the $m$-width by $1/5$ and the $n$-width by $1-n_0$, and split $m$ on a tie. Both children contain their shared midpoint boundary. To verify a path, start with~\eqref{eq:root} and repeat the prescribed splits, choosing the lower or upper child according to each bit. Check equality with the stored endpoints. The set of paths must be prefix-free, and every visited non-leaf prefix must have descendants in both children. This proves a full binary-tree cover by induction from its leaves to the root. A separate supplemental audit rechecks this combinatorial cover using the prefix property and exact Kraft equality $\sum_{p}2^{-|p|}=1$~\cite[Chapter 5]{CoverThomas2006}; it does not reimplement the interval classifier.

The comparison BSC is selected by 20 dyadic bisections on $[0,1]$, retaining a proposed lower endpoint only if the upper bound for its information is strictly less than $R$. Verification checks this strict condition directly. The method requires a certified lower-information channel, not an exact inverse of $J$.

\subsection{Enclosure of elementary operations}
The endpoint rules below implement outwardly rounded, inclusion-preserving interval arithmetic~\cite[Sections 3.2 and 4.3]{MooreKearfottCloud2009}.
Write $S_0=2^{96}$. Every interval is stored as integer endpoints $[l,h]$, representing $[l/S_0,h/S_0]$. Enclose a rational $a/b$ by rounding $S_0a/b$ down and up. Addition and subtraction use endpoint addition and sign reversal. Multiplication takes the minimum and maximum of the four endpoint products, divides by $S_0$, and rounds outward. Division, allowed only if the denominator interval excludes zero, takes the minimum and maximum of the four rational endpoint quotients $S_0u/v$ and rounds outward. This also handles negative denominators. For a nonnegative interval and integer exponent $e\geq1$, compute $l^e/S_0^{e-1}$ and $h^e/S_0^{e-1}$ with downward and upward integer rounding. The zeroth power is exactly one. Square-root endpoints are the lower and upper integer square roots of $lS_0,hS_0$. These are direct enclosures of monotone functions or extrema on a closed rectangle.

To evaluate the logarithm of a positive dyadic input $x$, choose an integer $e$ with $x=2^ey$ and $1\leq y\leq2$. The normalization here is exact for the dyadic inputs used by the entropy routine. Set $z=(y-1)/(y+1)$. The logarithm series~\cite[Eq. (4.6.4)]{NISTDLMF} and the geometric-tail bound proved in Appendix~\ref{app:elementary}, with 32 retained terms, give
\begin{equation}\label{eq:logtail}
\ln y=2\sum_{j=0}^{31}\frac{z^{2j+1}}{2j+1}+\rho,\qquad
0\leq\rho\leq\frac{2z^{65}}{65(1-z^2)}.
\end{equation}
All terms are evaluated with interval operations, and the upper remainder endpoint is added to the upper sum endpoint. The geometric bound is valid for $0\leq z<1$; thus the tiny outward enlargement past $1/3$ in a dyadic interval remains valid. The same routine encloses $\ln2$ at $z=1/3$. Then use $\ln x=e\ln2+\ln y$, with the usual signed interval multiplication for negative $e$.

Binary entropy is evaluated as $[-p\ln p-(1-p)\ln(1-p)]/\ln2$ at dyadic $0<p<1/2$, with exact values at zero and one half. Its monotonicity on $[0,1/2]$ gives an enclosure over an interval. The definitions $J(d)=1-h_2((1-d)/2)$ and $f(u)=J(\sqrt u)$ from Section~\ref{sec:bounds} preserve enclosure with the corresponding endpoint reversal, an application of inclusion-preserving interval evaluation~\cite[Sections 4.3 and 5.3]{MooreKearfottCloud2009}. Before evaluating $f$, its argument interval may be intersected with $[0,1]$ only because each exact argument used here is known to lie in that domain. Besides powers of moments, the only additional argument is $(g-1)^2d_0^2q/(1+q)^2\in[0,1]$: under the specialization $g-1=2$, the inequality $(g-1)^2q\leq(1+q)^2$ follows from $(1-q)^2\geq0$.

The repetition formula~\eqref{eq:Rformula} has positive denominators and uses only these rational interval operations. A half-integer power of $1-q_-$ in~\eqref{eq:boxdefs} uses an integer power and, if needed, one square root. The maximum $G_\tau$ in~\eqref{eq:Gbound} is computed exactly by the endpoint and stationary-point comparisons for that rational quadratic. Rule 4 is tested before this maximum to guarantee $\tau\leq1$. Thus both truncation and rounding errors are included in every leaf sign test.

\subsection{Reproduction and scope of software verification}\label{app:reproduction}
The certificates, verification programs, and detailed reproduction instructions, including per-degree counts and run reports, are available in the versioned GitHub supplement~\cite{MNBMSCertificates2026}. An independently implemented checker verifies all 3,599 leaf inequalities and complete coverage for $4\leq\ell\leq32$. The scalar reductions and the validity of the classification rules are proved in Sections~\ref{sec:potential}--\ref{sec:saturation}.

\clearpage
\begingroup
\small
\linespread{0.94}\selectfont
\interlinepenalty=10000
\bibliographystyle{IEEEtran}
\bibliography{references}
\endgroup
\end{document}